\documentclass[runningheads,envcountsame]{llncs}

\usepackage[T1]{fontenc}
\usepackage{graphicx}

\usepackage{quiver}
\usepackage{hyperref}
\usepackage{amsfonts,amsmath,amssymb}
\usepackage{mathtools}
\usepackage{enumitem}

\newcommand{\dom}{\text{dom}}
\newcommand{\cod}{\text{cod}}
\newcommand{\slice}[2]{#1 / #2}

\newcommand{\Tau}{\ensuremath{\mathcal{T}}}
\newcommand{\TT}{\Tau}
\newcommand{\CC}{\ensuremath{\mathcal{C}}}

\newcommand{\xx}{\overline{x}}

\newcommand{\hset}{\mathrm{hSet}}

\newcommand{\itSet}{\mathcal{V}}
\newcommand{\IRSet}{\mathcal{U}}

\newcommand{\Asm}{\ensuremath{\mathsf{Asm}}}

\newcommand{\op}{\ensuremath{\mathsf{op}}}
\newcommand{\iso}{\cong}
\newcommand{\opC}[1]{#1^{\op}}
\newcommand{\catOfEl}[1]{\int #1}

\newcommand{\Ty}{\ensuremath{\mathsf{Ty}}}
\newcommand{\Tm}{\ensuremath{\mathsf{Tm}}}

\newcommand{\cat}[1]{\ensuremath{\mathbf{#1}}}

\newcommand{\Fam}{\ensuremath{\cat{Fam}}}

\newcommand{\app}{\ensuremath{\mathsf{app}}}
\newcommand{\id}[1]{\ensuremath{\mathsf{id}_{#1}}}
\newcommand{\isodem}[1]{\ensuremath{i_{\mathsf{dem},{#1}}}}

\newcommand{\tmincompcat}[2]{\ensuremath{\mathsf{tm}_{#1}(#2)}}
\newcommand{\coerce}[2]{\ensuremath{#1 \uparrow_{#2}}}

\newcommand{\emphasize}[1]{\textbf{#1}}

\newcommand{\emptyctx}{[]}

\newcommand{\el}{\ensuremath{\mathsf{el}}}
\newcommand{\U}{\ensuremath{\mathsf{U}}}
\newcommand{\Uh}{\widehat{\U}}
\newcommand{\LiftAsm}[1]{#1_{\ensuremath{\mathsf{Asm}}}}
\newcommand{\UAsm}{\LiftAsm{\U}}
\newcommand{\UPsh}{\U_{\ensuremath{\mathsf{Psh}}}}
\newcommand{\PER}{\ensuremath{\mathsf{PER}}}

\newcommand{\constructedcwf}[2]{\mathsf{UnivMod}({#2})}

\newcommand{\cwffont}[1]{\ensuremath{\overline{#1}}}

\newcommand{\pair}[2]{\mathsf{pair}_{#1, #2}}

\newcommand{\PiCompCat}[2]{\ensuremath{\Pi(#1,#2)}}
\newcommand{\SigmaCompCat}[2]{\ensuremath{\Sigma(#1,#2)}}
\newcommand{\pairCompCat}[2]{\ensuremath{\mathsf{pair}_{{#1},{#2}}}}
\newcommand{\projCompCat}[2]{\ensuremath{\mathsf{proj}_{{#1},{#2}}}}
\newcommand{\appCompCat}[2]{\ensuremath{\app_{{{#1},{#2}}}}}
\newcommand{\lambdaCompCat}[2]{\ensuremath{\lambda_{{#1},{#2}}}}
\newcommand{\UnitCompCat}[1]{\ensuremath{\mathbf{1}_{#1}}}
\newcommand{\ttCompCat}[1]{\ensuremath{\mathsf{tt}_{#1}}}

\newcommand{\PiCwf}[2]{\ensuremath{\Pi(#1,#2)}}
\newcommand{\SigmaCwf}[2]{\ensuremath{\Sigma(#1,#2)}}
\newcommand{\pairCwf}[2]{\ensuremath{\mathsf{pair}_{{#1},{#2}}}}
\newcommand{\projoneCwf}[2]{\ensuremath{\mathsf{proj}^1_{{#1},{#2}}}}
\newcommand{\projtwoCwf}[2]{\ensuremath{\mathsf{proj}^2_{{#1},{#2}}}}
\newcommand{\appCwf}[2]{\ensuremath{\app_{{{#1},{#2}}}}}
\newcommand{\lambdaCwf}[2]{\ensuremath{\lambda_{{#1},{#2}}}}
\newcommand{\UnitCwf}[1]{\ensuremath{\mathbf{1}_{#1}}}
\newcommand{\ttCwf}[1]{\ensuremath{\mathsf{tt}_{#1}}}

\newcommand{\universeOperation}[1]{\ensuremath{{#1}^\U}}
\newcommand{\SigmaUniverse}[2]{\ensuremath{\universeOperation{\Sigma}}({#1},{#2})}
\newcommand{\PiUniverse}[2]{\ensuremath{\universeOperation{\Pi}}({{#1},{#2}})}
\newcommand{\UnitUniverse}[1]{\ensuremath{\universeOperation{\mathbf{1}}}_{#1}}

\newcommand{\Univiso}[2]{\ensuremath{i^{\el}_{{#1},{#2}}}}
\newcommand{\Uiso}[1]{\ensuremath{i^\U_{#1}}}

\newcommand{\cwfextention}[2]{\ensuremath{{#1}._{\mathsf{CwF}}{#2}}}
\newcommand{\cwfmorphism}{\ensuremath{\to_{\mathsf{CwF}}}}
\newcommand{\cwfsubst}[2]{\ensuremath{{#1[#2]_\mathsf{CwF}}}}

\usepackage[capitalize,nameinlink]{cleveref}  

\AddToHook{cmd/appendix/before}{\crefalias{section}{appendix}}

\spnewtheorem{notation}[theorem]{Notation}{\bfseries}{\rmfamily}
\spnewtheorem{construction}[theorem]{Construction}{\bfseries}{\rmfamily}
\spnewtheorem{intuition}[theorem]{Intuition}{\bfseries}{\rmfamily}
\spnewtheorem*{proofsketch}{Proof sketch}{\itshape}{\rmfamily}
\spnewtheorem*{mydef}{Definition}{\bfseries}{\itshape}

\crefname{notation}{Notation}{Notations}
\crefname{construction}{Construction}{Constructions}
\crefname{intuition}{Intuition}{Intuitions}

\newcommand{\coqdocbaseurl}{https://anonymous.4open.science/r/fossacs-strict-hott}
\newlength{\coqsymbolheight}
\newcommand{\coqsymbol}{\begingroup\normalfont\settoheight{\coqsymbolheight}{\{}\includegraphics[height=\coqsymbolheight]{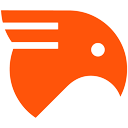}\endgroup}
\newcommand{\coqlink}{\begingroup\catcode`\#=12 \catcode`\%=12 \coqlinkaux}
\newcommand{\coqlinkaux}[1]{\endgroup\href{#1}{\coqsymbol}}

\newcommand{\allowhashin}[1]{\AddToHook{env/#1/begin}{\catcode`\#=12 }}
\allowhashin{theorem}
\allowhashin{lemma}
\allowhashin{proposition}
\allowhashin{corollary}
\allowhashin{definition}
\allowhashin{example}
\allowhashin{remark}
\allowhashin{notation}
\allowhashin{construction}

\begin{document}

\title{A General Construction of Strict Models in HoTT} 

\author{Niyousha Najmaei\inst{1} \and 
Niels van der Weide\inst{2}}
\institute{LIX, École Polytechnique, Inria, Palaiseau, France
\\
\email{najmaei@lix.polytechnique.fr}
\and
Eötvös Loránd University, Budapest, Hungary \\
\email{vanderweide@inf.elte.hu}
}

\authorrunning{N. Najmaei and N. Van der Weide} 

\maketitle

\begin{abstract}

Homotopy type theory (HoTT) provides a natural setting for studying the metatheory of type theory,
as it comes with additional extensionality principles compared to ordinary Martin-Löf type theory.
In particular, it is especially suitable for developing semantics of type theory,
where one distinguishes two kinds of models:
strict models,
which give us sound interpretations,
and pseudomodels.
Pseudomodels are useful as they occur prevalently
and can be turned into strict models by constructions called strictifications.
However,
the commonly used strictifications do not translate well to HoTT,
as they require assumptions that are generally not satisfied by concrete examples of models.

We show how to carry out the universe model construction, due to Voevodsky, in HoTT.
Under this construction, every pseudomodel together with a universe type closed under enough type formers
gives rise to a strict model.
We also show how to interpret type formers in the resulting model
assuming that the universe is closed under them.
Finally, we demonstrate that this construction generalises existing instances of strict set models in the literature,
and we apply it to more examples, yielding presheaf models and assembly models of type theory in HoTT.
Our results are formalised in Rocq using the UniMath library.
\end{abstract}

\section{Introduction}

Categorical semantics provides a 
framework for studying desirable properties of type theories, such as consistency and normalisation
\cite{DBLP:journals/lmcs/AltenkirchK17,coquand:2019,DBLP:conf/lics/Gratzer22,Kapulkin2021}.
A prominent example of a type theory whose categorical models are well-studied,
is Martin-L\"of type theory (MLTT).
MLTT serves as a foundation for constructive mathematics \cite{mltt}, and 
of many proof assistants, such as Rocq \cite{team:2024}, Lean \cite{demoura:2015} and Agda \cite{agda}.

One can study the semantics of type theory within type-theoretical foundations, 
which has numerous advantages.
One of the immediate advantages of such an approach is that it enables formalisation and direct
mechanisation of proofs in a proof assistant.
Another advantage is that type theoretical foundations provide better tools for 
reasoning about models by, for example, 
providing stronger and more suitable notions of invariance. 

Homotopy Type Theory (HoTT) is, for instance, a particularly suitable ambient theory for the study of 
the semantics of type theory.
Compared to ordinary MLTT,
HoTT has more extensionality principles,
namely the univalence axiom,
which states that types are identified if they are isomorphic.
The univalence axiom is beneficial in the formalisation of mathematics,
as it allows us to directly conclude that all properties of types are stable under isomorphism
and that we can transport proofs and structures across isomorphisms.
Without univalence, such properties are laborious to establish,
as they must be proved for each instance of interest.
In addition,
univalence can be instantiated to a wide variety of mathematical structures,
making its advantages broadly available~\cite{ahrens:2025}.

In this paper, we discuss a general technique for constructing models of type theory in HoTT.
We give an analysis on how this approach generalises various existing set models of 
type theory in the literature~\cite{altenkirch:2016,Gratzer_Gylterud_Mörtberg_Stenholm_2024} 
and also provide new examples, namely presheaf models and assembly models of type theory
in HoTT.

\subsubsection*{Semantics of Type Theories}

Various categorical structures have been studied for the semantics of type theory. 
We distinguish between two classes of such structures: strict models vs pseudomodels. 
The distinction lies in how substitution is treated: strict models interpret substitution strictly functorially (i.e., up to equality),
whereas pseudomodels interpret it only pseudofunctorially (i.e., up to isomorphism).
Although pseudomodels arise naturally in mathematics,
only strict models can interpret type theory, as substitution is usually strictly functorial there.
Examples of strict models include contextual categories~\cite{cartmell:1978}, categories with families (CwFs)~\cite{dybjer96} 
and natural models~\cite{awodey:2018a}, and examples of pseudomodels include 
display map categories~\cite{taylor87} and comprehension categories~\cite{jacobs93,jacobs99}.
In the remainder of this paper,
we use CwFs as our notion of strict model
and comprehension categories as our notion of pseudomodel.

Even though we cannot interpret type theory soundly in pseudomodels,
general methods, called strictification constructions, turn them into equivalent strict models. 
Various examples of such constructions have been used to interpret 
type theory \cite{hofmann:1994,gambino:2023,awodey:2009,10.1145/2754931},
most notably the ones based on 
the left and right adjoint splitting
constructions due to Giraud \cite{Giraud1965} and Bénabou \cite{streicher18}.
We give a more detailed discussion on strictification constructions in \cref{sec:setting}.

\subsubsection*{Semantics of Type Theories in HoTT}

When studying the semantics of type theories, defining the syntax and its models
presents subtle challenges in HoTT.
One approach to defining the syntax within type theory itself
was described by Altenkirch and Kaposi~\cite{altenkirch:2016}, who define the syntax of type theory using
\emph{quotient inductive-inductive types}~\cite{10.1145/3290315} as the
initial category with families.
Although successful, this approach admits only a restrictive class of models in HoTT.
For instance,
both the set and presheaf models of type theory,
only give rise to a CwF if one assumes \emph{uniqueness of identity proofs}---an assumption stating that all proofs of identity are equal, which is incompatible with the univalence axiom.
In response, Altenkirch and Kaposi~\cite{altenkirch:2016} use an inductive-recursive universe to get 
a set model.
Later Gratzer, Gylterud, M\"ortberg, and Stenholm~\cite{Gratzer_Gylterud_Mörtberg_Stenholm_2024} 
construct an alternative model using iterative sets.
However, neither gives a general technique for constructing models of type theory in HoTT.

We present a general technique for constructing strict models of type theory in HoTT.
We carry out the universe model construction, due to Voevodsky \cite{voevodsky2015csystemdefineduniversecategory}, in HoTT and 
show that it
serves as a general construction for strict models.
This construction generalises
the models by Altenkirch and Kaposi, and by Gratzer, Gylterud, M\"ortberg, and Stenholm.
We also instantiate it to more examples, yielding new strict presheaf models
and assembly models of type theory in HoTT.
Realizability models, such as assembly models, provide a versatile technique for 
constructing models of system F and the calculus of constructions~~\cite{luo:1990,reus:1999a,streicher:1991},
while presheaf models are used, among others, 
to prove normalisation by evaluation~\cite{DBLP:journals/lmcs/AltenkirchK17,coquand:2019,DBLP:conf/lics/Gratzer22}.
Our approach is appealing for these examples because
universes in presheaves and assemblies are well understood;
the challenge we tackle is showing that 
the construction can be carried out in HoTT for these examples.

\subsubsection*{Constructing Strict Models in HoTT}
In this paper, we carry out a universe model construction,
based on Voevodsky's \cite{voevodsky2015csystemdefineduniversecategory},
in HoTT.
We start with a comprehension category with a universe closed under unit and $\Sigma$-types
and construct a CwF by restricting to the universe.
Recall that comprehension categories are pseudomodels, whereas CwFs are strict models.

A universe in a comprehension category gives a type $\U_\Gamma$ in each context $\Gamma$ and a 
function $\el$ mapping each term $a$ of type $\U_\Gamma$ to a type $\el(a)$ in context $\Gamma.$
The construction of a CwF from the universe works as follows:
\begin{itemize}
  \item we take the CwF-contexts to be terms of the universe in the empty context;
  \item we take the CwF-types in a context $\Gamma$ to be terms of $\U_\Gamma$;
  \item for a CwF-type $A$, we take the CwF-terms of $A$ to be terms of $\el(A)$.
\end{itemize}

This yields a CwF, as the collections of terms of a certain type
in a comprehension category make a set, and substitution of terms is strictly functorial there.
Our construction is similar to Voevodsky's construction of C-systems from categories with a universe \cite{voevodsky2015csystemdefineduniversecategory}, 
with differences which we discuss in \Cref{sec:relwork}.
We also show that
the resulting CwF supports unit and $\Sigma$-types, 
supports $\Pi$-types if the starting universe is closed under $\Pi$-types, and is democratic.

\subsubsection*{Contributions}
The contributions of this paper can be summarised as follows. 
\begin{enumerate}
    \item We carry out a universe model construction, inspired by Voevodsky's \cite{voevodsky2015csystemdefineduniversecategory}, 
    in HoTT foundations: a general construction for strict models in HoTT that can be used to 
    develop semantics of type theory within HoTT (\cref{construction,thm:construction});
    \item We show that $\Sigma$-, unit, and $\Pi$-types are preserved under this construction (\cref{thm:universe_sigma,thm:universe_unit,thm:universe_pi}) and the resulting CwFs are democratic (\cref{thm:democracy});
    \item We show that this construction generalises two existing examples of set models of type theory in HoTT 
    (\cref{exa:iterative,exa:ir-univ});
    \item We present new examples by instantiating the construction to obtain presheaf and assembly models of type theory in HoTT (\cref{exa:assembly,exa:psh}).
\end{enumerate}

The main results of the paper are formalised in Rocq in the UniMath library \cite{unimath} and
are available \href{\coqdocbaseurl}{here}. We also provide links to specific sections of the 
formalisation throughout the paper using the Rocq logo \coqsymbol. 
Our formalisation in \Cref{sec:example_ass} builds on the work by Speight and Van der Weide~\cite{speight:2026}.

\subsubsection*{Synopsis}
In \cref{sec:setting}, we give a more precise description of the setting of the problem and discuss various strictification constructions.
In \cref{sec:background}, we review the background on CwFs and comprehension categories. 
In \cref{sec:universe_comp_cat}, we give the definition of universes of comprehension categories and
 their closure properties under type formers. 
In \cref{sec:construction}, we present the construction from a comprehension category with a universe to
a category with families. 
We also show that the construction preserves certain type formers, namely $\Sigma$-, $\Pi$- and unit type and that the resulting CwF is democratic.
In \cref{sec:examples}, we instantiate the construction to set models, presheaf models and assembly models.
In \cref{sec:relwork}, we discuss related work and compare our construction to various existing ones. 

\section{Setting}
\label{sec:setting}
In this section, we give a more detailed explanation of the problem addressed in this paper,
which is the construction of strict models in HoTT.
We review the commonly used strictification constructions,
namely the ones based on the left and right adjoint splitting constructions~\cite{Giraud1965,streicher18} 
and argue that they are not suitable for HoTT.
Recall that there are two kinds of models:
strict models, in which substitution is modelled strictly functorially,
and pseudomodels, which model substitution pseudofunctorially (functorial up to isomorphism).

Before we look at strictification constructions,
we explicate our notions of strict models and pseudomodels.
We use categories with families (CwF) as our notion of strict model.
Every CwF comes with an interpretation of type theory,
since the syntax of type theory can be defined as the initial CwF~\cite{altenkirch:2016}.
Our notion of pseudomodel should be general enough to admit interesting examples in HoTT,
such as set and presheaf models.
As our notion of pseudomodel, we take comprehension categories
since they are a more general notion
compared to other commonly used pseudomodels of type theory \cite{DBLP:conf/aplas/AhrensLN24} encompassing more 
examples of constructive models of type theory \cite{gambino:2023,10.1145/3776725}.
In addition,
every univalent category with finite limits gives rise to a comprehension category~\cite{vanderweide2025internallanguagesunivalentcategories},
so this notion encapsulates the set and presheaf models.

The notion of homotopy level becomes subtly important when discussing strictification constructions in HoTT.
Recall that a type $X$ is called a set if for all $x, y : X$ and $p, q : x = y$ we have $p = q$.
We denote the type of sets by $\hset$.
Not every type is a set,
since,
for instance,
$\hset$ itself is not a set.
A crucial distinction between our notions of strict models and pseudomodels lies in requirements on homotopy levels.
In CwFs,
the types in a certain context are required to be a set
and substitution is strict,
making them a suitable class of strict models.
However,
we relax these requirements for pseudomodels,
for which we use comprehension categories.
In general,
the types in a comprehension category do not form a set.
It is important to relax this requirement for pseudomodels,
because,
for instance,
in set models of type theory,
where types are interpreted as members of $\hset$,
the types do not form a set themselves.

To build a strict model from a pseudomodel in HoTT,
we construct a CwF from a comprehension category.
This requires altering the comprehension category in two ways.
First,
we must replace the collection of types by a set.
Second,
we must change the substitution operation on types
so that it becomes strictly functorial.
Note that in set-theoretic foundations, such constructions only need to alter the substitution operation.

To understand how comprehension categories can be strictified, one needs to know how type theory is interpreted there. 
In comprehension categories, the category of types is fibred over the category of contexts.
Although fibrations can in general only admit a pseudofunctorial interpretation of substitution, there is a class of fibrations,
where strictly functorial substitution can be interpreted: \emph{split} fibrations.
Strictification constructions for fibration-based models are based on constructions that
build split fibrations equivalent to the starting non-split ones. 

There are two main such splitting constructions, the left and right adjoint splitting constructions 
due to Giraud \cite{Giraud1965} and Bénabou \cite{streicher18}, where 
the splittings are given as left and right adjoints to the inclusion functor
$i : \mathsf{SplitFib}(\CC) \hookrightarrow \mathsf{Fib}(\CC)$.
The right adjoint splitting was used by Hofmann to interpret Martin-L\"of type theory in locally Cartesian closed categories~\cite{hofmann:1994},
and one can also use the right adjoint splitting to interpret homotopy type theory using algebraic weak factorisation systems~\cite{gambino:2023}.
Lumsdaine and Warren showed how to use the left adjoint splitting to interpret type theory using model categories~\cite{awodey:2009,10.1145/2754931}.

For a complete explanation of the left and right adjoint splittings, we refer the reader to \cite{streicher18} and \cite{CURIEN201499,hofmann:1994} respectively. 
We also do not explicate how these splitting constructions lift from
fibrations to models of type theory, e.g. comprehension categories. 
For this, we refer the reader to \cite{CURIEN201499,10.1145/2754931}.
We do, however, discuss whether they are a suitable for HoTT.

Let us check whether the left and right adjoint splitting constructions are suitable for a strictification construction in HoTT.
Since, in fibration-based models, types are interpreted as objects in the fibres,
this amounts to verifying that the fibres of the resulting split fibration form a set.

Let $p : \TT \to \CC$ be a (cloven) fibration and let $\Gamma : \CC$.
In the right adjoint splitting,
objects over $\Gamma$ consist of $A : \TT_\Gamma$ together with a function
mapping a morphism $f : \Gamma' \to \Gamma$ to a cartesian lifting of $f$ into $A$ such that the function maps the identity morphism to the identity lift.
We denote this type by $\TT_R$.
In the left adjoint splitting, objects over $\Gamma$ are triples  $(U_A : \CC, E_A : \TT_U , [A] : \Gamma \to U_A )$,
and we denote this type of objects by $\TT_L$.

\begin{proposition}
If both $\TT$ and $\CC$ have a set of objects,
then $\TT_R$ and $\TT_L$ are sets.
\end{proposition}

Note that for $\TT_R$ and $\TT_L$ to be sets, we do need to require the objects of $\TT$ and $\CC$ to be sets. 
This is because identities of both $\TT_R$ and $\TT_L$ consist,
among others,
of identities of objects in $\TT$.

We can thus see that neither of these constructions satisfies our desired properties,
since if we start with a general pseudomodel,
the types in the resulting model are not guaranteed to form a set.
In particular, there might exist non-trivial proofs of identity in $\CC$ and $\TT$.

There is, however, an interesting special case of the left adjoint splitting, which does give us a strict model.
This is the universe model construction, due to Voevodsky~\cite{voevodsky2015csystemdefineduniversecategory}, 
where one first fixes a $U: \CC$ and $E: \TT_U$, and the fiber over 
$\Gamma$ becomes maps $\Gamma \to U$.
These fibers form a set,
since the morphisms in $\CC$ always form a set by definition.
In the remainder of this paper,
we study the universe model construction in the context of homotopy type theory.

\begin{remark}
We see the base category as the category of contexts and the fibers as types in that context.
The fixed data $U$ and $E$ can be thought of as 
having a universe $U$ and an $\el$ map in the sense of a Tarski style 
universe (see \cref{def:comp_cat_universe}). The splitting construction 
can be thought of as having the types in the resulting 
model be the terms of type $U$ in the starting model (see \cref{construction}).
\end{remark}

\begin{remark}
Before ending this section, it is worth discussing another special case of the left adjoint splitting,
which has a similar flavour to Voevodsky's universe model: the local universe model construction of Lumsdaine and Warren~\cite{10.1145/2754931}.
In this construction, they lift the requirement of having one fixed universe, 
and give conditions for 
which the strictification works with a weaker notion of stability of type formers. 
The local universe
model is not applicable for our case, as we are interested in building CwFs where type formers are
strictly stable.
\end{remark}

\section{Models of Type Theory}
\label{sec:background}
In this section, we review categories with 
families~\cite{DBLP:journals/corr/abs-1904-00827,dybjer96} 
and comprehension categories~\cite{jacobs93,jacobs99}. We also give intuition on the interpretation of type 
theory.  
For an introduction to fibrations, we refer the reader to the notes by Streicher~\cite{streicher18}.
 
\subsection{Categories with Families}
Categories with families (CwFs) are a class of models of type theory, where types and terms are interpreted as a presheaf over a 
category of contexts $\CC$, i.e. a functor
$T : \CC^{\op} \to \Fam(\hset)$.
CwFs are strict models meaning that types form a set and substitution is modelled strictly functorially.

In HoTT, the requirement that types and terms form a set is especially noteworthy.
As a consequence, one cannot, for example, define a CwF whose contexts are sets,
and whose types in a context $\Gamma$ are
functions $\Gamma \to \hset$, without assuming uniqueness of identity proofs.
This is because, for a $\Gamma: \hset$, the type of functions $\Gamma \to \hset$ is not a set in general.
Besides set models, this obstacle also applies to presheaf models, another important class of examples.

\begin{definition}[\cite{dybjer96},\coqlink{https://anonymous.4open.science/r/fossacs-strict-hott/src/CategoriesWithFamilies/CatsWithFams.v#L464}]
    \label{def:cwf}
    A \emphasize{category with families} (CwF) consists of:
    \begin{enumerate}
        \item a category $\CC$ with a terminal object;

        \item a functor $T : \CC^{\op} \to \Fam(\hset)$ mapping each $\Gamma : \CC$ to the set family 
        $(\Tm_{\Gamma}(A))_{A : \Ty(\Gamma)}$;
        we denote $T(s) (A,t)$ by $(A[s], t[s])$ for each $s : \Gamma \to \Delta$ and appropriate $A$ and $t$;

        \item a context comprehension operation which assigns to each $\Gamma : \CC$ and $A : \Ty(\Gamma)$ a tuple 
        $(\Gamma.A : \CC$, $p_A : \Gamma.A \to \Gamma, q_A : \Tm_{\Gamma.A}(A[p]))$, satisfying the following universal property: 
        for any $\Delta : \CC$, $s : \Delta \to \Gamma$ and $t : \Tm_\Delta(A[s])$, there exists a unique 
        $u : \Delta \to \Gamma.A$ such that $p_A \circ u = s$ and $t = q_A[u]$. 
        We denote this unique $u$ by $\langle s, t \rangle$.
    \end{enumerate}
\end{definition}

\begin{intuition}
    Objects of the category $\CC$ can be seen as contexts and morphisms as substitutions.
    The terminal object in $\CC$ gives the empty context. 
    The family $(\Tm_\Gamma(A))_{A : \Ty(\Gamma)}$ can be thought of as a family of terms indexed by types
    in the context $\Gamma$. Note that this is a family of sets, meaning that we have $\Ty(\Gamma) : \hset$  
    and $\Tm_\Gamma(A) : \hset$.
    For each $s : \Gamma \to \Delta$ in $\CC$,
    the two components of $T(s)$ give substitution of types and terms as follows.
    For each $A : \Ty(\Delta)$, $A[s] : \Ty(\Gamma)$
    is the application of the first component of $s$ to $A$, and for each $t : \Tm_\Delta(A)$, $t[s] : \Tm_\Gamma(A[s])$ is
    the application of the second component of $s$ to $t$.
    This means that substitution is strictly functorial for both types and terms by definition.
    The context $\Gamma.A$ can be thought of as $\Gamma$ extended by the type $A$, and $p_A$ as the weakening substitution.
     The term $q_A$ gives the variable term. The universal property can be understood as 
    substitutions being built out of terms in the usual sense, i.e. substitutions $s': \Gamma \to \Delta.A$ 
    are tuples $(s : \Gamma \to \Delta,t : \Tm_\Gamma(A[s]))$.
\end{intuition}

The type former structures on a CwF are discussed at length by Castellan, Clairambault, and Dybjer \cite{DBLP:journals/corr/abs-1904-00827}.
For brevity, we only sketch these definitions here; and refer to \cite{DBLP:journals/corr/abs-1904-00827} for full definitions.
Our focus is on unit, $\Pi$-, $\Sigma$-types,
and we also consider democratic CwFs.

\begin{definition}[{\cite[Definition~28]{DBLP:journals/corr/abs-1904-00827}},
    \coqlink{https://anonymous.4open.science/r/fossacs-strict-hott/src/CategoriesWithFamilies/CatsWithFams.v#L979}]
    A \emphasize{$\Pi$-type structure on a CwF} consists of, for each $\Gamma : \CC$, $A : \Ty(\Gamma)$ and $B : \Ty(\Gamma.A)$,
    a type $\PiCwf{A}{B} : \Ty(\Gamma)$, a function $\lambdaCwf{A}{B} : \Tm_{\Gamma.A}(B) \to \Tm_{\Gamma}(\PiCwf{A}{B})$
    and a function $\appCwf{A}{B} : \Tm_{\Gamma}(\PiCwf{A}{B}) \to \prod_{a : \Tm_{\Gamma}(A)}  \Tm_{\Gamma}(B[\langle \id{\Gamma}, a \rangle])$
    such that 
    $\lambda$ and $\app$ induce a bijection
    and $\Pi$, $\lambda$ and $\app$ are preserved under substitution.
\end{definition}

\begin{definition}[{\cite[Definition~26]{DBLP:journals/corr/abs-1904-00827}},
    \coqlink{https://anonymous.4open.science/r/fossacs-strict-hott/src/CategoriesWithFamilies/CatsWithFams.v#L1159}]
    A \emphasize{$\Sigma$-type structure on a CwF} consists of, for each $\Gamma : \CC$ $A : \Ty(\Gamma)$ and $B : \Ty(\Gamma.A)$,
    a type $\SigmaCwf{A}{B} : \Ty(\Gamma)$,
    projection functions 
    $
    \projoneCwf{A}{B} : \Tm_{\Gamma}(\SigmaCwf{A}{B}) \to \Tm_{\Gamma}(A)
    $
    and
    \[
    \projtwoCwf{A}{B} : \prod_{p : \Tm_{\Gamma}(\SigmaCwf{A}{B})} 
    \Tm_{\Gamma}\big(B[\langle \id{\Gamma},\projoneCwf{A}{B}(p)\rangle]\big),
    \]
    and a pairing function $\pairCwf{A}{B} : \prod_{a : \Tm_{\Gamma}(A)} \Tm_{\Gamma}(B[\langle \id{\Gamma}, a \rangle]) \to \Tm_{\Gamma}(\SigmaCwf{A}{B})$
    such that
    the projection and pairing functions induce a bijection, and $\Sigma$,
    $\mathsf{proj}^1$, $\mathsf{proj}^2$ and $\mathsf{pair}$ are preserved under substitution.
\end{definition}

\begin{definition}[\coqlink{https://anonymous.4open.science/r/fossacs-strict-hott/src/CategoriesWithFamilies/CatsWithFams.v#L845}]
A \emphasize{(strong) unit type structure on a CwF} consists of, for each
$\Gamma : \mathcal{C}$, a type $\UnitCwf{\Gamma}$ and a term $\ttCwf{\Gamma} : \Tm_\Gamma(\UnitCwf{\Gamma})$ 
such that $\ttCwf{\Gamma}$ is unique and for every $s : \Gamma \to \Delta$ in $\mathcal{C}$, we have
$\UnitCwf{\Delta}[s] = \UnitCwf{\Gamma}$.
\end{definition}

\begin{definition}[{\cite[Definition~3]{DBLP:journals/corr/abs-1904-00827}},\coqlink{https://anonymous.4open.science/r/fossacs-strict-hott/src/CategoriesWithFamilies/CatsWithFams.v#L1293}]
    A CwF is \emphasize{democratic} if for each $\Gamma : \CC$, there exists a type $\overline{\Gamma} : \Ty(\emptyctx)$ and an isomorphism $\isodem{\Gamma}:\Gamma \iso \emptyctx.\overline{\Gamma}$ in $\CC$.
\end{definition}

\subsection{Comprehension Categories}
Comprehension categories are a class of pseudomodels of type theory where substitution of types is modelled using a fibration of types over contexts. 
In particular, comprehension categories allow for the interpretation of pseudofunctorial substitution in general.
In HoTT, it is important to note that types in a comprehension category are generally not a set, whereas terms are, and substitution of terms in comprehension categories is strictly functorial.
In addition, the notion of comprehension category is more flexible than CwFs: it admits various examples in HoTT,
like the set model.
Below we present the usual definition of comprehension categories
as given, for example, by Jacobs~\cite[Definition~4.1]{jacobs93}.
In type theory, however, one would phrase the definition slightly differently,
formulating fibrations in terms of displayed categories~\cite[Definition 6.1]{AL19};
our formalisation is based on the displayed categorical variant.
Since we do not use any specifics of displayed categories in this paper,
we disregard these aspects.

\begin{definition}[{\cite[Definition~4.1]{jacobs93}}] \label{def:comp_cat}
A \emphasize{comprehension category} consists of a category $\CC$ with a terminal object $\emptyctx$, a fibration $p : \TT \to \CC$, and a functor $\chi : \TT \to \CC^\to$ preserving cartesian arrows, such that the following diagram commutes.
\[
\begin{tikzcd}
\TT \arrow[rr, "\chi"] \arrow[dr, "p" swap] && \CC^{\to} \arrow[dl, "\cod"] \\
& \CC
\end{tikzcd}
\]    
A comprehension category is \emphasize{split} if $p$ is split and \emphasize{full} if $\chi$ is fully faithful.
\end{definition}

\begin{intuition} 
    The objects of the category $\CC$ can be thought of as contexts, and the terminal object of $\CC$ as the empty context. 
For $\Gamma : \CC$, the objects in the fiber over $\Gamma$ can be thought of as types in that context.
Note that the objects of the categories $\CC$ and $\TT$ are not necessarily sets. 
This means that types in a comprehension category are not set-truncated. 
The reindexing functors given by the fibration model pseudofunctorial substitution (see \cref{sec:universe_comp_cat}).
The functor $\chi$ models context extension in the following sense.
For each $A : \TT_\Gamma$, we think of $\chi(A)$ as the projection $\pi_A : \Gamma.A \to \Gamma$.
Terms of the type theory are interpreted as sections of such projections; hence, terms of a given type form a set 
as the morphisms in $\CC$ form a set.
The preservation of cartesian morphisms under $\chi$ makes the interpretation of substitution of terms possible (see \cref{not:comp_cats}).
\end{intuition}

\begin{remark}
    The original definition of a comprehension category does not require $\CC$ to have terminals, i.e. empty context.
    Since our universes (see \cref{sec:universe_comp_cat}) use the empty context to yield a type in every context, we add this assumption here.
\end{remark}

Before continuing, we set some notation for comprehension categories aligned with how 
type-theoretical constructs are interpreted in comprehension categories~\cite{jacobs93,10.1145/3776725}. 
Note that \cref{notation:coercion} is less used in the literature. 
\begin{notation}
We introduce the following notation.
  \label{not:comp_cats}
    \begin{enumerate}
        \item We use $\pi_A$ to denote $\chi(A)$ and $\Gamma.A$ to denote $\dom(\chi(A))$, for $A : \TT_\Gamma$.
        \item We use $\tmincompcat{\Gamma}{A}$ to denote the sections of $\chi(A) : \Gamma.A \to \Gamma$ for $A : \TT_\Gamma$.
        \item We use $A[s]$ to denote the reindexing of $A : \TT_\Delta$ along $s : \Gamma \to \Delta$ in $\CC$.
        \item We use $t[s]$ to denote the morphism given by the universal property of the following pullback square given by $\chi$ preserving cartesian morphisms: \label{item:pullback}
  \[\begin{tikzcd}
	\Gamma & \Delta \\
	& {\Gamma.s^* A} & {\Delta.A} \\
	& \Gamma & \Delta
	\arrow["s", from=1-1, to=1-2]
	\arrow["{{t[s]}}"{description}, dashed, from=1-1, to=2-2]
	\arrow[curve={height=12pt}, equals, from=1-1, to=3-2]
	\arrow["t", from=1-2, to=2-3]
	\arrow["{{s.A}}", from=2-2, to=2-3]
	\arrow["{{\chi (s^* A)}}"{description}, from=2-2, to=3-2]
	\arrow["\lrcorner"{anchor=center, pos=0.125}, draw=none, from=2-2, to=3-3]
	\arrow["{{\chi A}}"{description}, from=2-3, to=3-3]
	\arrow["s"', from=3-2, to=3-3]
  \end{tikzcd}\]
  for a section $t$ of $\chi A$, $A : \TT_\Delta$ and $s : \Gamma \to \Delta$. 
        \item \label{notation:coercion} We use $\coerce{t}{a}$ to denote $\Gamma.a \circ t$ for each $a : A \to A'$ in $\TT_\Gamma$ and $t : \tmincompcat{\Gamma}{A}$, which is itself in $\tmincompcat{\Gamma}{A'}$.
        This notation is inspired by the literature on interpreting coercive subtyping as 
        vertical morphisms in $\TT$~\cite{coraglia_et_al:LIPIcs.TYPES.2023.3,10.1145/3776725}.
        Under this view, given a coercion witness $a : A \to A'$, one can coerce a term $t$ of type $A$ 
        to a term $\Gamma.a \circ t$ of type $A'$.
    \end{enumerate}
\end{notation}

For $\Sigma$- and $\Pi$-types, we use the definitions by Lumsdaine and Warren \cite{10.1145/2754931}; these correspond to their \emph{pseudo stable} $\Sigma$- and $\Pi$-types structures. For brevity, we give informal definitions 
here and refer the reader to \cite{10.1145/2754931} for full definitions.
In the following definitions, let $(\CC, \TT, p, \chi, \emptyctx)$ be a comprehension category.

\begin{definition}[{\cite[Definitions~3.4.2.1, 3.4.2.8]{10.1145/2754931}},
    \coqlink{https://anonymous.4open.science/r/fossacs-strict-hott/src/ComprehensionCats/CompCatTypeFormers.v#L436}] \label{def:pi-non-full}
    Given $\Gamma : \CC$, $A : \TT_\Gamma$ and $B : \TT_{\Gamma.A}$, a 
    \emphasize{$\Pi$-type structure} consists of an object $\PiCompCat{A}{B} : \TT_\Gamma$,
    a morphism $\appCompCat{A}{B} : \Gamma.A.\PiCompCat{A}{B}[\pi_A] \to \Gamma.A.B$ such that 
    $\pi_B \circ \appCompCat{A}{B} = \pi_{{\PiCompCat{A}{B}[\pi_A]}}$, and 
    a function $\lambdaCompCat{A}{B}$ giving for each section $b: \Gamma.A \to \Gamma.A.B$ of $\pi_B$, a section $\lambdaCompCat{A}{B} (b) : \Gamma \to \Gamma.\PiCompCat{A}{B}$ of $\pi_{\PiCompCat{A}{B}}$
    such that $\lambdaCompCat{A}{B}(-)$ and $\appCompCat{A}{B} \circ (-)[\pi_A]$ establish a bijection, and $\Pi$, $\mathsf{app}$ and $\lambda$ are preserved under substitution.
\end{definition}

\begin{definition}[{\cite[Definition~3.4.4.1]{10.1145/2754931}},
    \coqlink{https://anonymous.4open.science/r/fossacs-strict-hott/src/ComprehensionCats/CompCatTypeFormers.v#L109}] \label{def:sigma-non-full}
    Given $\Gamma : \CC$, $A : \TT_\Gamma$ and $B : \TT_{\Gamma.A}$, \emphasize{(strong) $\Sigma$-type structure} consists of an object $\SigmaCompCat{A}{B} : \TT_\Gamma$ and 
        an isomorphism $\pairCompCat{A}{B} : \Gamma.A.B \to \Gamma.\SigmaCompCat{A}{B}$ with inverse $\projCompCat{A}{B}$ 
      such that $\pi_{\SigmaCompCat{A}{B}} \circ \pairCompCat{A}{B} = \pi_A \circ \pi_B$, and 
$\Sigma$ and $\mathsf{proj}$ are preserved under substitution. 
\end{definition}

\begin{definition}[\coqlink{https://anonymous.4open.science/r/fossacs-strict-hott/src/ComprehensionCats/CompCatTypeFormers.v#L523}]
    \label{def:comp_cat_unit}
   Let $\Gamma : \CC$.
A \emphasize{(strong) unit type structure} in context $\Gamma$ consists of a type $\UnitCompCat{\Gamma} : \TT_\Gamma$ and
a unique term $\ttCompCat{\Gamma} : \tmincompcat{\Gamma}{\UnitCompCat{\Gamma}},$
i.e. we have $u = \ttCompCat{\Gamma}$  for every term $u : \tmincompcat{\Gamma}{\UnitCompCat{\Gamma}},$
such that $\UnitCompCat{-}$ is stable under substitution, i.e. we have
    $
    i_{\mathbf{1},s} : \UnitCompCat{\Delta}[s] \iso \UnitCompCat{\Gamma}
    $
    in $\TT_\Gamma$ for each $s : \Gamma \to \Delta$ in $\CC$.
\end{definition}

\begin{remark}
Unlike for $\Sigma$- and $\Pi$-types, our definition for the unit differs from that of
Lumsdaine and Warren~\cite{10.1145/2754931}.
They define unit types via an elimination and a computation rule,
whereas we use strong unit types with a unique term.
Uniqueness of this term is needed for having an empty (terminal) context 
in the resulting CwF of \cref{construction}.
Our unit structure induces theirs. 
\end{remark}

\section{Universes in Comprehension Categories}
\label{sec:universe_comp_cat}
The main ingredient of the universe model construction is a universe that contains the unit type
and is closed under $\Sigma$-types.
We set up this needed machinery by
defining universes of comprehension categories and their closure properties.
We also define when a universe contains $\Pi$-types,
so that we can reason about $\Pi$-types in the resulting model of the construction.
We use the definition of a universe in a comprehension category given by Van der Weide \cite[Definition~9.13]{vanderweide2025internallanguagesunivalentcategories}. 
A similar definition for categories with families can be found in Angiuli and Gratzer's book \cite[Structure~6.4.17]{angiuli2025principles}.

In this section, let $(\CC,\TT,p,\chi,\emptyctx)$ be a comprehension category. 

\begin{definition}[{\cite[Definition~9.13]{vanderweide2025internallanguagesunivalentcategories}},
  \coqlink{https://anonymous.4open.science/r/fossacs-strict-hott/src/ComprehensionCats/CompCatUniverse.v#L317}]
    \label{def:comp_cat_universe}
    A \emphasize{universe structure} on $\chi$ consists of:
\begin{enumerate}
\item an object $\U : \TT_{\emptyctx}$, which gives rise to $\U_\Gamma \coloneq \U[\iota_\Gamma]$ in each $\TT_\Gamma$ where $\iota_\Gamma : \Gamma \to \emptyctx$ is the terminal morphism;
\item for each $\Gamma : \CC$, a function $el : \tmincompcat{\Gamma}{\U_\Gamma} \to \TT_\Gamma$; \label{item:el}
\item \label{item:universe_subst}
for each $s : \Gamma \to \Delta$ and $t : \tmincompcat{\Delta}{\U_\Delta}$, an isomorphism
$
\Univiso{s}{t} : (\el (t))[s] \iso \el (\coerce{t[s]}{\Uiso{s}})
$ \label{item:isomorphism_universe}
stating that reindexing commutes with $\el$,
where $\Uiso{s} : \U_{\Delta}[s] \iso \U_\Gamma$ is the isomorphism 
$\U_{\emptyctx}[\iota_\Delta][s] \iso \U_{\emptyctx}[\iota_\Gamma]$
given by the universal property of the terminal object $\emptyctx$;
\end{enumerate}
such that the universe is coherent under substitution as stated in \cref{sec:coherences}.
\end{definition}

One could equivalently define $\el$ as a type in context 
$\emptyctx.\U$ and obtain the function in \cref{item:el} by
reindexing along the canonical morphism $\Gamma  \to \emptyctx.\U$.
We choose \cref{def:comp_cat_universe},
as it is closer to the type theoretic presentation of universes.

 \begin{notation}
    \label{notation:Utermpulledback}
    For $s : \Gamma \to \Delta$ in $\CC$, we denote by $\Uiso{s}$ the isomorphism $\U_{\Delta}[s] \iso \U_\Gamma$ from
    \cref{item:isomorphism_universe} of \cref{def:comp_cat_universe}.
\end{notation}

We are now ready to define the closure properties of the universe under type formers. We define what it means for a universe to be 
closed under $\Sigma$-, $\Pi$- and unit types. 
Let us first give an intuition of what these closure properties should look like. 
For $\Sigma$-types, for example,
we say that whenever we have terms $a$ and $b$ of type $\U$, their $\Sigma$-type also lies in $\U$.
Specifically,
for each
$a,b : \tmincompcat{}{\U}$ with $b$ depending on $a$, we have $\SigmaUniverse{a}{b}$ in $\U$.
These $\Sigma$-codes should also be preserved under substitution.
Additionally, we want to 
have an isomorphism $\el(\SigmaUniverse{a}{b}) \iso \SigmaCompCat{\el(a)}{\el(b)}$, also stable under substitution.
Similar closedness definitions are given by Van der Weide~\cite{vanderweide2025internallanguagesunivalentcategories} 
for categories with finite limits. 

For brevity, we state the coherences of \cref{def:comp_cat_univ_closed_sigma,def:comp_cat_univ_closed_pi} in \cref{sec:coherences}.
\begin{definition}[\coqlink{https://anonymous.4open.science/r/fossacs-strict-hott/src/ComprehensionCats/CompCatUniverse.v#L1416}]
\label{def:comp_cat_univ_closed_sigma}
Let $(\CC,\TT,p,\chi, \emptyctx)$ support $\Sigma$-types.
A universe $(\U,\el)$ is said to be \emphasize{closed under $\Sigma$-types}
if for each $\Gamma : \CC$, $a : \tmincompcat{\Gamma}{\U_\Gamma}$ and $b : \tmincompcat{\Gamma.\el(a)}{\U_{\Gamma.\el(a)}}$,
we have a term $\SigmaUniverse{a}{b} : \tmincompcat{\Gamma}{\U_\Gamma}$ and an isomorphism 
$
i_{\SigmaUniverse{a}{b}} : \el(\SigmaUniverse{a}{b}) \iso \SigmaCompCat{\el(a)}{\el(b)},
$
where $\SigmaCompCat{\el(a)}{\el(b)}$ is the $\Sigma$-type in the comprehension category, such that 
reindexing commutes with $\Sigma$-codes, i.e. for each $s : \Delta \to \Gamma$ in $\CC$, we have 
    \[ \coerce{\SigmaUniverse{a}{b}[s]}{\Uiso{s}} = \SigmaUniverse{\coerce{a[s]}{\Uiso{s}}}{\coerce{b[s.\el(a) \circ \Delta.i^{-1}_{s,a}]}
    {\Uiso{s.\el(a)}}},\]
    and $i_{\SigmaUniverse{a}{b}}$ is coherent under reindexing as stated in \cref{sec:coherences}.
\end{definition}

We define closure under $\Pi$-types analogously as follows.

\begin{definition}[\coqlink{https://anonymous.4open.science/r/fossacs-strict-hott/src/ComprehensionCats/CompCatUniverse.v#L1877}]
\label{def:comp_cat_univ_closed_pi}
Let $(\CC,\TT,p,\chi, \emptyctx)$ support $\Pi$-types.
A universe $(\U,\el)$ is said to be \emphasize{closed under $\Pi$-types}
if for each $\Gamma : \CC$, $a : \tmincompcat{\Gamma}{\U_\Gamma}$ and $b : \tmincompcat{\Gamma.\el(a)}{\U_{\Gamma.\el(a)}}$,
we have a term $\PiUniverse{a}{b} : \tmincompcat{\Gamma}{\U_\Gamma}$ and an isomorphism 
$
i_{\PiUniverse{a}{b}} : \el(\PiUniverse{a}{b}) \iso \PiCompCat{\el(a)}{\el(b)},
$
where $\PiCompCat{\el(a)}{\el(b)}$ is the $\Pi$-type in the comprehension category,
such that  reindexing commutes with $\Pi$-codes, i.e. for each $s : \Delta \to \Gamma$ in $\CC$, we have 
    \[ \coerce{\PiUniverse{a}{b}[s]}{\Uiso{s}} = \PiUniverse{\coerce{a[s]}{\Uiso{s}}}{\coerce{b[s.\el(a) \circ \Delta.i^{-1}_{s,a}]}
    {\Uiso{s.\el(a)}}},\]
    and $i_{\PiUniverse{a}{b}}$ is coherent under reindexing as stated in \cref{sec:coherences}.
\end{definition}

Closure for unit types is similar to the closure for $\Sigma$- and $\Pi$-Types, 
but preservations under substitution are derivable in this case.

\begin{definition}[\coqlink{https://anonymous.4open.science/r/fossacs-strict-hott/src/ComprehensionCats/CompCatUniverse.v#L1630}]
    \label{def:comp_cat_univ_closed_unit}
    Let $(\CC,\TT,p,\chi,\emptyctx)$ support unit types.
    A universe $(\U,\el)$ is said to be \emphasize{closed under unit types}
    if we have a term $\UnitUniverse{} : \tmincompcat{\emptyctx}{\U_{\emptyctx}}$, which induces a term $\UnitUniverse{\Gamma} : \tmincompcat{\Gamma}{\U_\Gamma}$ for each $\Gamma: \CC$, and
    an isomorphism 
    $ i_{\UnitUniverse{}}: \el(\UnitUniverse{}) \iso \UnitCompCat{\emptyctx},
    $
    where $\UnitCompCat{\emptyctx}$ is the unit in the comprehension category.
\end{definition}

\begin{proposition}[\coqlink{https://anonymous.4open.science/r/fossacs-strict-hott/src/ComprehensionCats/CompCatUniverse.v#L1668} and 
    \coqlink{https://anonymous.4open.science/r/fossacs-strict-hott/src/ComprehensionCats/CompCatUniverse.v#L1785}]
    \label{prop:unit-subst}
    The unit codes defined in \cref{def:comp_cat_univ_closed_unit} are stable under substitution and $\el(\UnitUniverse{})$ has a unique term.     
\end{proposition}
\begin{proof}
    It follows from uniqueness of $\ttCompCat{}$.
\end{proof}

\section{The Universe Model Construction}
\label{sec:construction}
In this section, we present our universe model construction, inspired by that of Voevodsky \cite{voevodsky2015csystemdefineduniversecategory}:
we construct a CwF from a comprehension category with a universe closed under unit and $\Sigma$-types. 
In short, being closed under unit ensures having an empty context in the resulting CwF and being closed
under $\Sigma$ gives context extension.
We then show that the resulting CwF is democratic and the construction preserves
type formers. In particular, the resulting CwF comes with $\Sigma$- and unit types 
using the closedness properties 
of the universe under these type formers.
We also show that if the universe is closed under $\Pi$-types, the resulting CwF
has $\Pi$-types that mimic the $\Pi$-types structure of the universe. 

In the rest of this section,  let $(\CC,\TT,p,\chi,\emptyctx)$ be a comprehension category with a universe $(\U,\el)$ 
closed under unit and $\Sigma$-types.

\begin{construction}[\coqlink{https://anonymous.4open.science/r/fossacs-strict-hott/src/ComprehensionCats/CwfFromCompCatWithUniv.v#L445}]
    \label{construction}
    We construct a CwF $\constructedcwf{\CC}{\U}$ as follows:
    \begin{itemize}
    \item the category of contexts has $\tmincompcat{\emptyctx}{\U}$ as objects, and morphisms are context morphisms of the form $\emptyctx.\el(\Gamma) \to \emptyctx.\el(\Delta)$ in the comprehension category;
    \item the terminal context of the CwF is then given by $\UnitUniverse{}$, the code of the unit in the universe;
    \item types in a CwF-context $\Gamma : \tmincompcat{\emptyctx}{\U}$ are $\tmincompcat{\emptyctx.\el(\Gamma)}{\U_{\emptyctx.\el(\Gamma)}}$ 
    in the comprehension category;
    \item terms of a CwF-type $A : \tmincompcat{\emptyctx.\el(\Gamma)}{\U_{\emptyctx.\el(\Gamma)}}$ in a CwF-context 
    $\Gamma : \tmincompcat{\emptyctx}{\U}$ are $\tmincompcat{\emptyctx.\el(\Gamma)}{\el(A)}$ in the comprehension category;
    \item substitution of both types and terms in the CwF is given by substitution of terms in the
    comprehension category in the following sense:
    for a CwF-type $A : \tmincompcat{\emptyctx.\el(\Gamma)}{\U_{\emptyctx.\el(\Gamma)}}$ 
    and a context morphism $s : \emptyctx.\el(\Delta) \to \emptyctx.\el(\Gamma)$, the 
    resulting substituted CwF-type is given by $\coerce{A[s]}{\Uiso{s}}$, where $A[s]$ denotes the
    substitution of terms in the comprehension category as defined in \cref{not:comp_cats} and
    $\Uiso{s}$ is the isomorphism explained in \cref{notation:Utermpulledback}.
    Similarly, for a CwF-term $t : \tmincompcat{\emptyctx.\el(\Gamma)}{\el(A)}$, the substituted term
    is  $\coerce{t[s]}{\Univiso{s}{t}}$;
    \item context extension follows from the universe being closed under $\Sigma$-types in the following sense.
    Given a CwF-context $\Gamma : \tmincompcat{\emptyctx}{\U}$ and a CwF-type $A : \tmincompcat{\emptyctx.\el(\Gamma)}{\U_{\emptyctx.\el(\Gamma)}}$, 
    the extended context $\cwfextention{\Gamma}{A}$ 
    is given by $\SigmaUniverse{\Gamma}{A}$.
    For projection and the term $q_A$, let us first introduce a morphism we call $\overline{\mathsf{proj}}$ defined as follows:
    \[\begin{tikzcd}[ampersand replacement=\&, column sep = large]
	    {\emptyctx.\el(\SigmaUniverse{\Gamma}{A})} \& {\emptyctx.\SigmaCompCat{{\el(\Gamma)}}{\el(A)}} \& {\emptyctx.\el(\Gamma).\el(A)}.
	    \arrow["{\emptyctx.i_{\SigmaUniverse{\Gamma}{A}}}", from=1-1, to=1-2]
	    \arrow["{\projCompCat{\el(\Gamma)}{\el(A)}}", from=1-2, to=1-3]
    \end{tikzcd}\]  
    The projection $p_{\cwfextention{\Gamma}{A}}$ is given by $\pi_{\el(A)} \circ \mathsf{proj}$ and $q_A$ is defined as $\coerce{v_{\el(A)} [\overline{\mathsf{proj}}]}{e}$ in the comprehension category,
    where $v_{\el(A)}$ is the variable term of the comprehension category, 
    and $e$ is the coherence isomorphism $(\el(A)[\pi_{\el(A)}])[\overline{\mathsf{proj}}] \iso
    \el(A [\pi_{\el(A)} \circ \overline{\mathsf{proj}}]) $.
    
\end{itemize}
\end{construction}

\begin{theorem}[\coqlink{https://anonymous.4open.science/r/fossacs-strict-hott/src/ComprehensionCats/CwfFromCompCatWithUniv.v#L1026}]
    \label{thm:construction}
    \cref{construction} yields a CwF $\constructedcwf{\CC}{\U}$.
\end{theorem}
\begin{proofsketch}
  We need to show the uniqueness of the extended substitution
  $\langle s, t\rangle$, which is defined as $s.A \circ t[s]$ using the analogous operation in comprehension categories.
  Uniqueness follows from the universal property of the pullback from \cref{item:pullback} of 
  \cref{not:comp_cats} and the coherence of $\U$ (see \cref{def:universe-coherence} in \cref{sec:coherences}).
\end{proofsketch}

Recall that our goal was to construct CwF models of type theory in HoTT without assuming types are set-truncated.
To see why \cref{construction} achieves this, we check two things: that $\constructedcwf{\CC}{\U}$ is 
a CwF and that our starting model has no set-truncation assumption on its types.
For the first, we check two things: that types in $\constructedcwf{\CC}{\U}$
form an $\hset$ and that $\constructedcwf{\CC}{\U}$ is a strict model.
Both follow from types of
$\constructedcwf{\CC}{\U}$ being defined as certain terms of the comprehension category, which, as explained in \cref{sec:background}, are defined as certain morphisms in the base category and form an $\hset$.
For the second point, comprehension categories are pseudomodels with no assumption
on set-truncation of types, and are general enough to admit a variety of examples such as set models, 
presheaf models and assembly models.

\begin{remark}
    \label{rem:fullness}
    When interpreting type theory in comprehension categories, 
    one usually assumes that the comprehension category is full.
    We do not make this assumption meaning that 
    \cref{construction} works for general comprehension categories,
    which are used to model more expressive type theories~\cite{adjedj:hal-05167997,coraglia_et_al:LIPIcs.TYPES.2023.3,10.1145/3776725}.
\end{remark}

Before proceeding, we point out the similarity between this construction and the \emph{externalisation} of internal categories.
Specifically,
every internal category in a finite complete category $\CC$ gives rise to a split (Grothendieck) fibration over $\CC$.
Since universes give rise to full internal categories,
they give rise to a CwF via externalisation.
However, note that the universe model construction is formulated in more generality:
rather than only considering finitely complete categories,
we consider all comprehension categories with a universe.

The interpretation of higher-order logic in a topos~\cite{lambek:1988} uses similar ideas as \cref{construction}.
Every topos has a subobject classifier,
giving a universe of propositions~\cite{maietti:2005},
that contains the unit type (i.e., the truth proposition),
and is closed under $\Sigma$-types (i.e., existential quantification).
Propositions are interpreted as terms of the subobject classifier
rather than as monomorphisms.

\subsubsection*{Type Formers and Democracy} 
We now show that $\constructedcwf{\CC}{\U}$ is democratic and supports unit and $\Sigma$-types, and if $\U$ is closed under $\Pi$-types, it also supports $\Pi$-types.
In the remainder of this section, let $(\CC,\TT,p,\chi)$ be a comprehension category with 
a universe $(\U,\el)$ closed under unit and $\Sigma$-types. 
We use the shorthands $\cwfsubst{A}{s}$ and $\cwfsubst{t}{s}$ to 
denote the substitution of types and terms in $\constructedcwf{\CC}{\U}$ as explained in \cref{construction},
and $\cwffont{\UnitCwf{\Gamma}}$, $\cwffont{\PiCwf{A}{B}}$ and $\cwffont{\SigmaCwf{A}{B}}$ to 
denote unit, $\Sigma$ and $\Pi$ in $\constructedcwf{\CC}{\U}$.

\begin{theorem}[\coqlink{https://anonymous.4open.science/r/fossacs-strict-hott/src/ComprehensionCats/CwfFromCompCatWithUnivUnitDemocracy.v#L166}]
    \label{thm:democracy}
    $\constructedcwf{\CC}{\U}$ is democratic.
\end{theorem}
\begin{proofsketch}[Proof in \cref{sec:proofs}]
    For each context $\Gamma : \tmincompcat{\emptyctx}{\U}$, we have a canonical type $\bar{\Gamma}$ in the empty CwF-context 
    defined as follows in $\CC$:
    \[\begin{tikzcd}[ampersand replacement=\&,row sep=large]
        {\emptyctx.\el(\UnitUniverse{\emptyctx})} \& \emptyctx \\
        \& {\emptyctx.\el(\UnitUniverse{\emptyctx}).\U} \& {\emptyctx.\U} \\
        \& {\emptyctx.\el(\UnitUniverse{\emptyctx})} \& \emptyctx.
        \arrow["\pi_{\el(\UnitUniverse{\emptyctx})}", from=1-1, to=1-2]
        \arrow["{\bar{\Gamma}}", dashed, from=1-1, to=2-2]
        \arrow[curve={height=12pt}, equals, from=1-1, to=3-2]
        \arrow["\Gamma", from=1-2, to=2-3]
        \arrow[from=2-2, to=2-3]
        \arrow[from=2-2, to=3-2]
        \arrow["\lrcorner"{anchor=center, pos=0.125}, draw=none, from=2-2, to=3-3]
        \arrow[from=2-3, to=3-3]
        \arrow["\pi_{\el(\UnitUniverse{\emptyctx})}"', from=3-2, to=3-3]
    \end{tikzcd}\]
    The universal property of the pullback gives the required isomorphism.
\end{proofsketch}

\begin{remark}
\label{rem:democracy_remark}
In Voevodsky's original construction, a contextual category is built from a universe category without restricting
the category of contexts. We restrict it to obtain democratic CwFs (\cref{thm:democracy}),
as democratic CwFs with extensional identity types and $\Sigma$-types
are equivalent to categories with finite limits
and that equivalence can be extended to various notions of categories~\cite{claira11,vanderweide2025internallanguagesunivalentcategories}.
For axiomatic identity types, similar assumptions yield analogous equivalences~\cite{otten:2026}.
Without this restriction, the universe need not be closed under unit and $\Sigma$-types.
\end{remark}

The type former structures of $\constructedcwf{\CC}{\U}$ closely mimics the codes of 
the type formers of the comprehension category in the universe, as defined in 
\cref{def:comp_cat_univ_closed_sigma,def:comp_cat_univ_closed_pi,def:comp_cat_univ_closed_unit}.
We also have to show that the resulting type former structure is (strictly) 
stable under substitution.  
One of the obstacles in the following proofs, is the difference between how the 
type former structure is defined in a CwF compared to a comprehension category. 
Specifically, terms in a comprehension category are defined as sections of 
the projection morphisms of the form $\Gamma.A \to \Gamma$, whereas terms in a CwF are
a primitive notion. Another difference is that variable terms and extension 
of substitution are defined using the universal property of certain pullbacks 
in comprehension categories, whereas they are primitive in a CwF.
In our formalisation, we first bridge this gap by deriving the type former
structure of a CwF from one in which terms are presented as sections.
This intermediate step makes it easier to define the type formers for $\constructedcwf{\CC}{\U}$
from those of $\U$.
We do not explicate this translation here; for this we 
point the reader to our formalisation (see \coqlink{https://anonymous.4open.science/r/fossacs-strict-hott/src/CategoriesWithFamilies/CatsWithFamsTypeFormersWithSections.v}).

\begin{theorem}[\coqlink{https://anonymous.4open.science/r/fossacs-strict-hott/src/ComprehensionCats/CwfFromCompCatWithUnivUnitDemocracy.v#L95}]
   $\constructedcwf{\CC}{\U}$ has unit types.
    \label{thm:universe_unit}
\end{theorem}
\begin{proofsketch}[Proof in \cref{sec:proofs}]
    Let $\Gamma$ be a context in $\constructedcwf{\CC}{\U}$, i.e. $\Gamma: \tmincompcat{\emptyctx}{\U}$.
    The unit type $\cwffont{\UnitCwf{\Gamma}}$ in $\constructedcwf{\CC}{\U}$ is given by the unit code 
    $\UnitUniverse{\emptyctx.\el(\Gamma)} : \tmincompcat{\emptyctx.el(\Gamma)}{\U_{\emptyctx.el(\Gamma)}}$ in the original model.
    The unique unit term  $\cwffont{\ttCwf{\Gamma}}$ in $\constructedcwf{\CC}{\U}$ is given by the term $\ttCompCat{\emptyctx.\el(\Gamma)}$ coerced 
    along
    $i^{-1}_{\UnitUniverse{\emptyctx.\el(\Gamma)}}.$
\end{proofsketch}

\begin{theorem}[\coqlink{https://anonymous.4open.science/r/fossacs-strict-hott/src/ComprehensionCats/CwfFromCompCatWithUnivSigma.v#L1529}]
   $\constructedcwf{\CC}{\U}$ has $\Sigma$-types.
    \label{thm:universe_sigma}
\end{theorem}
\begin{proofsketch}[Proof in \cref{sec:proofs}]
  Let ${\Gamma}$ be a context in $\constructedcwf{\CC}{\U}$.
  Suppose that we have types ${A}$ in context ${\Gamma}$ 
  and ${B}$ in context
    ${\cwfextention{{\Gamma}}{{A}}}$ in $\constructedcwf{\CC}{\U}$.
    The type $\cwffont{\SigmaCwf{A}{B}}$ in $\constructedcwf{\CC}{\U}$ is given by the following $\Sigma$-code in $\U$:
    \[
    \cwffont{\SigmaCwf{A}{B}} \coloneqq \SigmaUniverse{{A}}{
    \cwfsubst{{B}}{\pair{\el({{\Gamma}})}{\el({A})} \circ i_{\SigmaUniverse{{\Gamma}}{{A}}}}}.
    \]
    Pairing and projection are defined using \cref{def:comp_cat_univ_closed_sigma},
    and preservation under substitution is proved using standard lemmas about comprehension
    categories and the coherence of $\Sigma$-codes given in \cref{def:sigma-coherence} in \cref{sec:coherences}.
\end{proofsketch} 

\begin{theorem}[\coqlink{https://anonymous.4open.science/r/fossacs-strict-hott/src/ComprehensionCats/CwfFromCompCatWithUnivPi.v#L269}]
    \label{thm:universe_pi}
    If $\U$ is closed under $\Pi$-types, $\constructedcwf{\CC}{\U}$ has $\Pi$-types.
\end{theorem}
\begin{proofsketch}[Proof in \cref{sec:proofs}]
     The type $\cwffont{\PiCwf{A}{B}}$ in $\constructedcwf{\CC}{\U}$ is 
     given by the following $\Pi$-code in the original model
     \[
     \cwffont{\PiCwf{A}{B}} \coloneqq \PiUniverse{{A}}{
     \cwfsubst{{B}}{\pair{\el({{\Gamma}})}{\el({A})} \circ i_{\SigmaUniverse{{\Gamma}}{{A}}}}}.
     \]
     for a suitable ${\Gamma}$, ${A}$ and ${B}$.
     Application and $\lambda$ are defined using \cref{def:comp_cat_univ_closed_pi},
     and preservation under substitution is proved using standard lemmas about comprehension
     categories and the coherence of $\Pi$-codes given in \cref{def:sigma-coherence} in \cref{sec:coherences}.
\end{proofsketch}

\section{Examples}
\label{sec:examples}

Finally,
we instantiate \cref{construction} to obtain strict models of type theory within HoTT.
We do so by equipping various models of type theory with a universe;
specifically, in the categories of sets, assemblies, and presheaves.
We see that our construction generalizes established 
strict models 
in the literature~\cite{altenkirch:2016,Gratzer_Gylterud_Mörtberg_Stenholm_2024},
and also give new examples,
namely assembly and presheaf models.

\subsection{Universes in Sets}
\label{sec:example_set}
To construct a universe in the category of sets,
we need to give a set $\U$ together with a map that assigns to terms $x : \U$ sets $\el(x)$.
We can deduce closure conditions of the resulting universe,
which we denote by $\Uh$,
from closure conditions on $\U$.
For instance,
if we have $u : \U$ such that $\el(u)$ is isomorphic to the unit type,
then $\Uh$ contains the unit type.
In addition,
$\Uh$ is closed under $\Sigma$-types
if for all $a : u$ and $b : \el(a) \rightarrow \U$,
we have $\sigma_{a, b} : \U$ such that $\el(\sigma_{a, b})$ and $\Sigma (x : \el(a)), \el(b(x))$ are isomorphic.
In our formalisation,
we use this general notion of universe (\coqlink{https://anonymous.4open.science/r/fossacs-strict-hott/src/SetModel/SetUniverses.v#L34}),
and we show that it gives rise to a universe in the comprehension category of sets
that is closed under the same type formers as $\U$.
This means that we can apply \Cref{construction} to this general notion of universe (\coqlink{https://anonymous.4open.science/r/fossacs-strict-hott/src/SetModel/SetUniverseTypes.v#L610})
to obtain strict models.
In the remainder,
we look at this construction for two concrete instantiations:
iterative sets~\cite{Gratzer_Gylterud_Mörtberg_Stenholm_2024} and induction-recursion~\cite{dybjer:2000}.

Iterative sets are based on the work by Aczel on relating constructive set theory and type theory~\cite{aczel:1978},
and later they were used by Gratzer, Gylterud, M\"ortberg, and Stenholm to construct the set model of type theory as a category with families~\cite{Gratzer_Gylterud_Mörtberg_Stenholm_2024}.
We do not recall the definition here, but we recall the main properties that we use.
Specifically,
we have a type $\itSet$ of iterative sets,
which is a set itself.
In addition, every $X : \itSet$ gives rise to a set $\el(X)$,
which allows us to conclude that $\itSet$ is a Tarski-style universe of sets.
Iterative sets are closed under $\Sigma$-types and $\Pi$-types,
and the unit type is an iterative set.
Hence, we get a universe $\itSet$,
which allows us to apply \Cref{construction}.

\begin{example}
\label{exa:iterative}
We have a CwF $\constructedcwf{\CC}{\itSet}$ of iterative sets.
Contexts are maps from the unit type to $\itSet$,
and hence, the same as iterative sets.
Given an iterative set $\Gamma$,
types in context $\Gamma$ are the same as maps $A : \el(\Gamma) \rightarrow \itSet$.
Given an iterative set $\Gamma$
and a type $A : \el(\Gamma) \rightarrow \itSet$,
terms of type $A$ 
are maps that assign to $x : \el(\Gamma)$ an element $\el(A(x))$.
This CwF is equivalent to the one described by Gratzer, Gylterud, M\"ortberg, and Stenholm~\cite{Gratzer_Gylterud_Mörtberg_Stenholm_2024}.
\end{example}

Another way to obtain a universe of sets is via induction-recursion~\cite{dybjer:2000,dybjer:1999},
which Altenkirch and Kaposi use to define the standard model~\cite{altenkirch:2016}.
We do not recall the definition of universes as inductive-recursive types~\cite{dybjer:2000,dybjer:1999,palmgren:1998}.

\begin{example}
\label{exa:ir-univ}
Every inductive-recursive universe $\IRSet$ induces a CwF $\constructedcwf{\CC}{\IRSet}$.
Contexts in $\constructedcwf{\CC}{\IRSet}$ are members of $\IRSet$,
types in context $\Gamma$ are functions $\el(\Gamma) \rightarrow \IRSet$,
and terms of type $A : \el(\Gamma) \rightarrow \IRSet$ in context $\Gamma$
are functions that assign to $x : \el(\Gamma)$ an inhabitant of $\el(A(x))$.
This model is equivalent to the one given by Altenkirch and Kaposi~\cite{altenkirch:2016}.
\end{example}

\subsection{Universes for Assemblies}
\label{sec:example_ass}
The next example comes from realizability models of type theory~\cite{luo:1990,reus:1999a,streicher:1991}.
These models come with an impredicative universe of propositions
and an impredicative universe of types.
Grothendieck universes in the metatheory can be used to construct predicative universes~\cite{luo:1990}.

To obtain a universe in the realizability model,
we use Luo's construction~\cite{luo:1990},
who constructs a universe $\UAsm$ of assemblies from a universe $\U$ of sets.
The carrier of $\UAsm$ is the type of all assemblies whose carrier lies in $\U$.
For instance,
the universe $\itSet$ of iterative sets gives rise to a universe $\LiftAsm{\itSet}$ of assemblies,
and we also obtain a universe $\LiftAsm{\IRSet}$
from any inductive-recursive universe $\IRSet$.
We can apply \Cref{construction} to both $\LiftAsm{\itSet}$ and $\LiftAsm{\IRSet}$.

\begin{example}
\label{exa:assembly}
\Cref{construction} gives us a CwF $\constructedcwf{\CC}{\LiftAsm{\itSet}}$.
Contexts in this CwF are equivalent to assemblies whose carrier is an iterative set,
and types in context $\Gamma$ are the same as families $A : \el(\Gamma) \rightarrow \Asm(A)$ of assemblies for $A$
such that each $A(x)$ is an iterative set.
Finally, given an assembly $\Gamma$
and a family $A$ of assemblies over $\Gamma$,
terms of type $A$ in context $\Gamma$
are maps that assign to $x : \el(\Gamma)$ an element $t(x) : \el(A(x))$.
The CwF $\constructedcwf{\CC}{\LiftAsm{\IRSet}}$ is similar.
\end{example}

\begin{remark}
Given a universe $\U$ of sets that is closed under $\Sigma$-types,
we want to show that its lift to assemblies is closed under $\Sigma$-types of assemblies as well.
However,
this requires showing that the isomorphism between $\el(\sigma_{a, b})$ and $\Sigma (x : \el(a)), \el(b(x))$ is tracked by 
an element of the PCA.
While we need to assume this for general universes,
it holds for the universes of iterative sets and inductive-recursive universes.
This is because in both cases, the given isomorphism is the identity function,
which can always be tracked.
The same holds for other type formers, like $\Pi$-types.
In the formalisation,
we only treat general universes, assuming that the resulting universe of assemblies is closed under $\Sigma$-types (\coqlink{https://anonymous.4open.science/r/fossacs-strict-hott/src/Realizability/RealizabilityCompCat.v#L1282}).
\end{remark}

\begin{remark}
Since assembly models come with a variety of universes,
\Cref{construction} can be applied to them in different ways.
In \Cref{exa:assembly}, we chose predicative universes,
but
one could use the impredicative universe $\PER$ of partial equivalence relations instead.
Applying \Cref{construction} to $\PER$, does not give a CwF with an impredicative universe of types,
as $\PER$ is not itself a $\PER$.
\end{remark}

\subsection{Universes for Presheaves}
\label{sec:example_presh}
Finally, we instantiate \Cref{construction} to presheaf models of type theory.
From a universe $\U$ of sets,
one can construct a universe of presheaves~\cite{awodey:2024,hofmann:1997b,slattery:2025}.
Specifically,
we define a presheaf $\UPsh$
which maps objects $c : \CC$
to the type of functors $\opC{(\slice{\CC}{c})} \rightarrow \U$.
We indeed have a set of such functors,
because $\U$ is a set.
To construct $\el$,
we note that from a presheaf $\Gamma : \opC{\CC} \rightarrow \hset$ 
and a natural transformation $\tau : \Gamma \Rightarrow \UPsh$
we get a functor $\el(\tau) : \opC{(\catOfEl{\Gamma})} \rightarrow \hset$,
which maps $x : \CC$ and $\xx$ to $\el(\tau(x, \xx)(\id{x}))$.
Both inductive-recursive universes and the universe of iterative sets induce a universe in presheaves.

If $\U$ contains the unit type and is closed under $\Sigma$-types,
then so is $\UPsh$.
We can thus use \Cref{construction} to obtain a CwF of presheaves.
However,
to show that $\UPsh$ is closed under $\Pi$-types,
we need to assume that $\U$ is closed under $\Pi$-types
and that the types of objects and of morphisms in $\CC$ lie in $\U$.

\begin{example}
\label{exa:psh}
By applying \Cref{construction} to the universe of the presheaf model,
we get a CwF $\constructedcwf{\CC}{\UPsh}$ such that
contexts are presheaves $\overline{\Gamma} : \opC{\CC} \rightarrow \U$
and types in context $\Gamma$ are presheaves $\opC{\catOfEl{\overline{\Gamma}}} \rightarrow \U$.
Terms and substitutions are the same as in the presheaf model.
\end{example}

\section{Related Work}
\label{sec:relwork}
Strictification constructions are prevalent in the literature on the semantics of type theory,
and the most commonly used are based on 
the right adjoint \cite{hofmann:1994,streicher18} and the left adjoint splitting constructions \cite{10.1145/2754931,Giraud1965}.
However,
as we argued in \Cref{sec:setting},
neither is suitable in HoTT
if one studies categories whose objects do not form a set.
We therefore considered a different construction.

Our construction is similar to Voevodsky's construction of C-systems from a universe category~\cite{voevodsky2015csystemdefineduniversecategory},
whose key idea likewise is to restrict a model to the types in a universe.
Compared to Voevodsky, we use a different starting point.
While we start with a comprehension category equipped with a universe,
the starting point of Voevodsky is a category $\CC$ with a terminal object and a universe structure.
Here, a universe structure consists of a morphism $p : \widetilde{U} \rightarrow U$
and a map that assigns to every $f : X \rightarrow U$ a pullback along $p$.
If we assume that $\CC$ has finite limits,
then one can directly construct a comprehension category with a universe from it.
However,
since Voevodsky does not assume that $\CC$ has all pullbacks,
he constructs a strict model of type theory as a C-system,
where types are defined in a slightly different way compared to us.
While we construct the collection of types in a single step and show that the resulting model is democratic~(\cref{rem:democracy_remark}),
Voevodsky uses an inductive construction to prove contextuality.
Voevodsky shows that the resulting C-system is closed under $\Pi$-types~\cite{voevodsky:2017a}.
We likewise show that the construction preserves $\Pi$-types, and further show that the resulting model is closed under $\Sigma$-types and unit types.

Larrea \cite{larrea:2018} gives an alternative presentation of Voevodsky's construction, more similar to ours, but 
with an important conceptual difference.
To equip the resulting model with $\Pi$- and $\Sigma$-types,
Larrea presents different closure conditions than we do.
We express closure conditions using the language of the comprehension category.
For instance, for closure under $\Sigma$-types,
we say that we have a term $\el(\SigmaUniverse{a}{b})$
and an isomorphism $\el(\SigmaUniverse{a}{b}) \iso \SigmaCompCat{\el(a)}{\el(b)}$.
Larrea expresses closure conditions with the internal language of the category of contexts,
and thus requires this category to be locally Cartesian closed.

Hua and Xu~\cite{hua2026polynomialfunctorspiclanssemantics} use a similar construction but for other purposes.
Just like Voevodsky, they obtain models that are not necessarily democratic,
whereas our construction has been modified to obtain democratic CwFs.
There are also foundational differences between our work and that of Hua and Xu.
Since we work in HoTT,
we must show that the collection of types in the constructed model forms a set.
Hua and Xu do not have to prove anything like that,
since they use Lean's type theory as foundation.
In addition,
they consider $\pi$-clans,
while we consider the more general notion of comprehension categories.
Compared to $\pi$-clans,
comprehension categories come with a notion of morphism between types,
which is used to interpret subtyping~\cite{coraglia_et_al:LIPIcs.TYPES.2023.3,10.1145/3776725}.
Since Hua and Xu use $\pi$-clans,
types in their models are represented by some property on the morphisms,
which states that there is some term in the universe that represents that type~\cite[Definition 5.5]{hua2026polynomialfunctorspiclanssemantics}.
In our model, types are interpreted more directly, as terms of the universe.

Our construction generalises the iterative sets model~\cite{Gratzer_Gylterud_Mörtberg_Stenholm_2024}
and the model based on inductive-recursive universes~\cite{altenkirch:2016},
as explained in \Cref{sec:examples}.
Extension of the iterative sets model to presheaf models have been considered by Grubm\"uller, Souche, and M\"ortberg~\cite{souche:2026,grubmuller:2026},
which is the model we considered in \Cref{sec:example_presh}.
An interesting aspect of their work lies in the used formalisation techniques.
They noted the necessity of using an abstract notion of universe to make a complete formalisation feasible.
Our work demonstrates that this technique  is widely applicable,
and we describe the general case.

An alternative approach to models of type theory in HoTT is based on the groupoid syntax of type theory described by Altenkirch, Kaposi, and Xie~\cite{altenkirch2025groupoidsyntaxtypetheoryset}.
They define the syntax of type theory as a 1-truncated higher inductive-inductive type rather than as a set-truncated type 
and prove that their syntax is a set.
Models of the 1-truncated syntax are known as groupoid categories with families (GCwF),
and one can show that sets form a GCwF.
Their approach presents an alternative way to obtain models of type theory in HoTT.
However, the cost is that one needs to identify suitable coherences and prove a coherence theorem.
Proving such a theorem can be challenging if there are, for instance, equations between types.
These problems do not occur when using the set-truncated syntax.
Damato used the notion of GCwF to present the container model~\cite{damato:2025a,vonglehn:2014}
and described how to get a strict version of the container model using induction-recursion.
We conjecture that the strict container model can be obtained as an instance of our construction.

\section{Conclusion}

In this paper, we carried out a universe model construction and showed that it is a general construction for strict models in HoTT.
More specifically, we started from the general notion of a comprehension category with a universe  and 
constructed a CwF by restricting to the types in the universe. 
We showed that this construction preserves several type formers, unit, $\Sigma$- and $\Pi$-types, and that
the resulting CwF is democratic. 
This construction serves as a general technique in HoTT, allowing us to obtain 
strict models such as set models, presheaf models and assembly models.
In particular, it generalises two existing set models of type theory, one based on a universe of 
iterative sets \cite{Gratzer_Gylterud_Mörtberg_Stenholm_2024} and  the other based on 
any inductive-recursive universe 
\cite{altenkirch:2016}. 
As for future work, we aim to study preservation of more type formers under \cref{construction},
e.g. identity types.

\subsubsection*{Acknowledgements}
The authors would like to thank Anders M\"ortberg for discussions regarding models of iterative sets
and Ambroise Lafont for useful discussions.
The authors would also like to thank Joseph Hua and Yiming Xu for discussions about the relation between our work.
This research was supported by the NWO project ``The Power of Equality'' OCENW.M20.380, which is financed by the Dutch Research Council (NWO).
This research was funded by the European Union (ERC, HOTT, 101170308).
Views and opinions expressed are however those of the author(s) only and do not
necessarily reflect those of the European Union or the European Research Council.
Neither the European Union nor the granting authority can be held responsible for them.
This work received government funding managed by the French National Research Agency under the France 2030 program, reference
“ANR-22-EXES-0013”.

\bibliographystyle{splncs04}
\bibliography{bibliography}

\appendix

\section{Full Definitions of Universe and Closure Properties}
\label{sec:coherences}

In this appendix, we give the complete definitions of a universe for comprehension categories
and the closure properties of the universe for $\Sigma$- and $\Pi$-types discussed 
in \cref{sec:universe_comp_cat},
including the coherences. 

Let us first define a useful notation.
\begin{notation}
    \label{not:elmap}
    In a comprehension category with a universe as defined in \cref{def:comp_cat_universe},
    let $t = t'$ be an equality of terms.
    We denote by 
    $i^{\mathrm{elmap}}$ the isomorphism $\el(t) \iso \el(t')$ given by path induction.
\end{notation}

\begin{definition}[{\cite[Definition~9.13]{vanderweide2025internallanguagesunivalentcategories}}]
    \label{def:universe-coherence}
Let $(\CC,\TT,p,\chi)$ be a comprehension category. A \emphasize{universe} for $(\CC,\TT,p,\chi)$ consists of:
\begin{enumerate}
\item An object $\U : \TT_{\emptyctx}$, which gives rise to $\U_\Gamma \coloneq \U[\iota_\Gamma]$ in each $\TT_\Gamma$ where $\iota_\Gamma : \Gamma \to \emptyctx$ is the terminal morphism;
\item For each $\Gamma : \CC$, a function $el : \tmincompcat{\Gamma}{\U_\Gamma} \to \TT_\Gamma$;
\item 
For each $s : \Gamma \to \Delta$ and $t : \tmincompcat{\Delta}{\U_\Delta}$, an isomorphism
$
\Univiso{s}{t} : (\el (t))[s] \iso \el (\coerce{t[s]}{\Uiso{s}})
$
stating that reindexing commutes with $\el$,
where $\Uiso{s} : \U_{\Delta}[s] \iso \U_\Gamma$ is the isomorphism 
$\U_{\emptyctx}[\iota_\Delta][s] \iso \U_{\emptyctx}[\iota_\Gamma]$
given by the universal property of the terminal object $\emptyctx$;
\end{enumerate}
such that
the following two diagrams commute:
\[
\begin{tikzcd}
  \el(t)
    \arrow[r, "i^{\mathrm{subst}}_\mathrm{id}"]
    \arrow[d, equals]
  & (\el(t))[\id{\Gamma}]
    \arrow[d, "\Univiso{\id{\Gamma}}{t}"] \\
  \el(t)
  & \el\big(\coerce{t[\id{\Gamma}]}{\Uiso{\id{\Gamma}}}\big)
    \arrow[l, "i^{\mathrm{elmap}}"']
\end{tikzcd}
\]
\[
\begin{tikzcd}
  (\el(t))[s_1][s_2]
    \arrow[r, "i^{\mathrm{subst}}_\mathrm{comp}"]
    \arrow[d, "{(\Univiso{s_1}{t})[s_2]}"']
  & (\el(t))[s_1 \circ s_2]
    \arrow[r, "\Univiso{s_1 \circ s_2}{t}"]
  & \el\big(\coerce{t[s_1 \circ s_2]}{\Uiso{s_1 \circ s_2}}\big)
    \arrow[d, "i^{\mathrm{elmap}}"] \\
  \big(\el(\coerce{t[s_1]}{\Uiso{s_1}})\big)[s_2]
    \arrow[rr, "\Univiso{s_2}{\coerce{t[s_1]}{\Uiso{s_1}}}"']
  &
  & \el\Big(\coerce{\big(\coerce{t[s_1]}{\Uiso{s_1}}\big)[s_2]}{\Uiso{s_2}}\Big)
\end{tikzcd}
\]
where $i^{\mathrm{subst}}_\mathrm{id}$ and $i^{\mathrm{subst}}_\mathrm{comp}$ are the usual coherence isomorphism for
reindexing along identity and composition in comprehension categories.
\end{definition}

In the next two definitions, we use the following two notations as shorthand:
\[
as \coloneqq \coerce{a[s]}{\Uiso{s}} \qquad bsa \coloneqq \coerce{b[s.\el(a) \circ \Gamma.i^{-1}_{s,a}]}
    {\Uiso{s.A}}
\]
for $s : \Delta \to \Gamma$ in $\CC$, $a : \tmincompcat{\Gamma}{\U_\Gamma}$ and $b : \tmincompcat{\Gamma.\el(a)}{\U_{\Gamma.\el(a)}}$.

\begin{definition}
 \label{def:sigma-coherence}
Let $(\CC,\TT,p,\chi)$ be a comprehension category supporting $\Sigma$-types.
A universe $(\U,\el)$ is said to be \emphasize{closed under $\Sigma$-types}
if for each $\Gamma : \CC$, $s : \Delta \to \Gamma$, $a : \tmincompcat{\Gamma}{\U_\Gamma}$ and $b : \tmincompcat{\Gamma.\el(a)}{\U_{\Gamma.\el(a)}}$,
we have:
\begin{enumerate}
\item a term $\SigmaUniverse{a}{b} : \tmincompcat{\Gamma}{\U_\Gamma}$;
\item an isomorphism 
$
i_{\SigmaUniverse{a}{b}} : \el(\SigmaUniverse{a}{b}) \iso \SigmaCompCat{\el(a)}{\el(b)},
$
where $\SigmaCompCat{\el(a)}{\el(b)}$ is the $\Sigma$-type in the comprehension category;
\end{enumerate}
such that 
\begin{enumerate}[resume]
    \item reindexing commutes with $\Sigma$-codes, i.e. for each $s : \Delta \to \Gamma$ in $\CC$, we have 
    \[ \coerce{\SigmaUniverse{a}{b}[s]}{\Uiso{s}} = \SigmaUniverse{as}{bsa},\]
  where $\Uiso{}$ is defined in \cref{notation:Utermpulledback}, and $as$ and $bsa$ are the notations defined above; 
    \item the isomorphism $i_{\SigmaUniverse{a}{b}}$ is coherent under reindexing, i.e. the following diagram commutes:
    \[
        \begin{tikzcd}
        \Gamma.\el(\SigmaUniverse{as}{bsa})
        \arrow[r, "i^{\Sigma}_{\Gamma,as,bsa}"]
        \arrow[d, "i^{\mathrm{elmap}}"']
        & \Gamma.\SigmaCompCat{\el(as)}{\el(bs)}
        \arrow[r, "i^{\mathrm{pair}}"]
        & \Gamma.\el(as).\el(bsa)
        \arrow[dd, "s.\el(a).\el(b)"] \\
        \Gamma.\el\big(\coerce{\SigmaUniverse{a}{b}[s]}{\Uiso{s}}\big)
        \arrow[d, "s.\el (\SigmaUniverse{a}{b})"']
        & &  \\
        \Delta.\el(\SigmaUniverse{a}{b})
        & \Delta.\SigmaCompCat{\el(A)}{\el(B)}
        \arrow[l, "(i_{\SigmaUniverse{a}{b}})^{-1}"]
        & \Delta.\el(a).\el(b)
        \arrow[l, "(i^{\mathrm{pair}})^{-1}"]
        \end{tikzcd}
    \]
    where $i^{\mathrm{pair}}$ is the pair-proj isomorphism of the $\Sigma$-types in a comprehension category
    defined in \cref{def:sigma-non-full}.
\end{enumerate}
\end{definition}

We define closure under $\Pi$-types analogously as follows.

\begin{definition}
 \label{def:pi-coherence}
Let $(\CC,\TT,p,\chi)$ be a comprehension category supporting $\Pi$-types.
A universe $(\U,\el)$ is said to be \emphasize{closed under $\Pi$-types}
if for each $\Gamma : \CC$, $a : \tmincompcat{\Gamma}{\U_\Gamma}$ and $b : \tmincompcat{\Gamma.\el(a)}{\U_{\Gamma.\el(a)}}$,
we have:
\begin{enumerate}
\item a term $\PiUniverse{a}{b} : \tmincompcat{\Gamma}{\U_\Gamma}$;
\item an isomorphism 
$
i_{\PiUniverse{a}{b}} : \el(\PiUniverse{a}{b}) \iso \PiCompCat{\el(a)}{\el(b)},
$
where $\PiCompCat{\el(a)}{\el(b)}$ is the $\Pi$-type in the comprehension category;
\end{enumerate}
such that 
\begin{enumerate}[resume]
    \item reindexing commutes with $\Pi$-codes, i.e. for each $s : \Delta \to \Gamma$ in $\CC$, we have 
    \[ \coerce{\PiUniverse{a}{b}[s]}{\Uiso{s}} = \PiUniverse{as}{bsa},\]
 where $\Uiso{}$ is defined in \cref{notation:Utermpulledback}, and $as$ and $bsa$ are the notations defined above;
    \item the isomorphism $i_{\PiUniverse{a}{b}}$ is coherent under reindexing, i.e. the following diagram commutes:
    \[
    \begin{tikzcd}
    \tmincompcat{\Delta.\el(a)}{\el(b)}
    \arrow[r, "\lambdaCompCat{a}{b}"]
    \arrow[d, "{\coerce{(-)[s.\el(a)]}{\Univiso{s.\el(a)}{b}}}"'] 
            & \tmincompcat{\Delta}{\PiCompCat{\el(a)}{\el(b)}}
    \arrow[d, "{\coerce{-}{(i_{\PiUniverse{a}{b}})^{-1}}}"] \\
    \tmincompcat{\Gamma.\el(as)}{\el(bsa)}
    \arrow[d, "\lambdaCompCat{\el(as)}{\el(bsa)}"']
            & \tmincompcat{\Delta}{\el(\PiUniverse{a}{b})}
    \arrow[d, "{(-)[s]}"] \\
    \tmincompcat{\Gamma}{\PiCompCat{\el(as)}{\el(bsa)}}
    \arrow[d, "{\coerce{-}{(i_{\PiUniverse{as}{bsa}} )^{-1}}}"']
            & \tmincompcat{\Gamma}{(\el(\PiUniverse{a}{b}))[s]}
    \arrow[d, "{\coerce{-}{\Univiso{s}{\PiUniverse{a}{b}}}}"] \\
    \tmincompcat{\Gamma}{\el(\PiUniverse{as}{bsa})}
            & \tmincompcat{\Gamma}{\el(\coerce{\PiUniverse{a}{b}[s]}{\Uiso{s}})}
    \arrow[l, "{\coerce{-}{i^{\mathrm{elmap}}}}"]
    \end{tikzcd}
    \]
    where $\lambdaCompCat{-}{-}$ is the usual lambda operation of comprehension categories defined  in \cref{def:pi-non-full}.
\end{enumerate}
\end{definition}

\section{Proofs of \Cref{sec:construction}}
\label{sec:proofs}

\begin{proof}[Proof of \Cref{thm:democracy}]
    For each context $\Gamma : \tmincompcat{\emptyctx}{\U}$, we have a canonical type $\bar{\Gamma}$ in the empty CwF-context 
    defined as follows in $\CC$:
    \[\begin{tikzcd}[ampersand replacement=\&,row sep=large]
        {\emptyctx.\el(\UnitUniverse{\emptyctx})} \& \emptyctx \\
        \& {\emptyctx.\el(\UnitUniverse{\emptyctx}).\U} \& {\emptyctx.\U} \\
        \& {\emptyctx.\el(\UnitUniverse{\emptyctx})} \& \emptyctx.
        \arrow["\pi_{\el(\UnitUniverse{\emptyctx})}", from=1-1, to=1-2]
        \arrow["{\bar{\Gamma}}", dashed, from=1-1, to=2-2]
        \arrow[curve={height=12pt}, equals, from=1-1, to=3-2]
        \arrow["\Gamma", from=1-2, to=2-3]
        \arrow[from=2-2, to=2-3]
        \arrow[from=2-2, to=3-2]
        \arrow["\lrcorner"{anchor=center, pos=0.125}, draw=none, from=2-2, to=3-3]
        \arrow[from=2-3, to=3-3]
        \arrow["\pi_{\el(\UnitUniverse{\emptyctx})}"', from=3-2, to=3-3]
    \end{tikzcd}\]
    We have to show $\Gamma \iso \cwfextention{\emptyctx}{\bar{\Gamma}}$, which corresponds to
    $\Gamma \iso \SigmaUniverse{\UnitUniverse{\emptyctx}}{\cwfsubst{\Gamma}{\pi_{\el(\UnitUniverse{\emptyctx})}}}$ in $\CC$.
    One can easily show that for this, it suffices to give an isomorphism of the form 
    $\emptyctx . \el (\Gamma) \iso \emptyctx . \el (\SigmaUniverse{\UnitUniverse{\emptyctx}}{\cwfsubst{\Gamma}{\pi_{\el(\UnitUniverse{\emptyctx})}}})$ in $\CC$, which follows from \cref{prop:unit-subst}.
\end{proof}

\begin{proof}[Proof of \Cref{thm:universe_unit}]
    Let $\Gamma$ be a context in $\constructedcwf{\CC}{\U}$, i.e. $\Gamma: \tmincompcat{\emptyctx}{\U}$.
    The unit type $\cwffont{\UnitCwf{\Gamma}}$ in $\constructedcwf{\CC}{\U}$ is given by the unit code 
    $\UnitUniverse{\emptyctx.\el(\Gamma)} : \tmincompcat{\emptyctx.el(\Gamma)}{\U_{\emptyctx.el(\Gamma)}}$ in the original model.
    The unique unit term  $\cwffont{\ttCwf{\Gamma}}$ in $\constructedcwf{\CC}{\U}$ is given by the term $\ttCompCat{\emptyctx.\el(\Gamma)}$ coerced 
    along
    $i^{-1}_{\UnitUniverse{\emptyctx.\el(\Gamma)}}$, where $\ttCompCat{}$ is the unique 
    term of unit in the comprehension category and $i^{-1}_{\UnitUniverse{\emptyctx.\el(\Gamma)}}$ is 
    the one defined in \cref{def:comp_cat_univ_closed_unit}.
    For the preservation under substitution, we have to show that 
    $\coerce{\cwffont{\UnitCwf{\Delta}} [s]}{\Uiso{s}} = \cwffont{\UnitCwf{\Gamma}}$
    for any $s : \emptyctx.\Gamma \to \emptyctx.\Delta$.
    This follows directly from $\cwffont{\UnitCwf{-}}$ being a term
    of type $\U_{\emptyctx.\el(-)}$.
\end{proof}

\begin{proof}[Proof of \Cref{thm:universe_sigma}]
  Let ${\Gamma}$ be a context in $\constructedcwf{\CC}{\U}$.
  Suppose that we have types ${A}$ in context ${\Gamma}$ 
  and ${B}$ in context
    ${\cwfextention{{\Gamma}}{{A}}}$ in $\constructedcwf{\CC}{\U}$.
    The type $\cwffont{\SigmaCwf{A}{B}}$ in $\constructedcwf{\CC}{\U}$ is given by the following $\Sigma$-code in the original model
    \[
    \cwffont{\SigmaCwf{A}{B}} \coloneqq \SigmaUniverse{{A}}{
    \cwfsubst{{B}}{\pair{\el({{\Gamma}})}{\el({A})} \circ i_{\SigmaUniverse{{\Gamma}}{{A}}}}}.
    \]
    The pair map $
    \cwffont{\pairCwf{A}{B}} : \cwfextention{\cwfextention{{\Gamma}}{{A}}}{{B}} \cwfmorphism \cwfextention{{\Gamma}}{\cwffont{\SigmaCwf{A}{B}}}.$
    is of the following form in the comprehension category:
    \[
    \emptyctx.\el(\SigmaUniverse{\SigmaUniverse{{\Gamma}}{{A}}}{{B}}) \to 
    \emptyctx.\el(\SigmaUniverse{{\Gamma}}{\SigmaUniverse{{A}}{\cwfsubst{{B}}{\pair{\el({{\Gamma}})}{\el({A})} \circ i_{\SigmaUniverse{{\Gamma}}{{A}}}}}}).
    \]
    This map is defined as a composition of instances of the isomorphism $i_{\SigmaUniverse{-}{-}}$, which pushes an $\el$ into the arguments 
    of $\SigmaUniverse{-}{-}$, and standard isomorphisms about $\Sigma$ in comprehension categories.
    The projection map is defined similarly, as a composite
    of different instances of the same isomorphisms, but in the opposite
    direction.
    It is then straightforward to prove the beta and eta rules. 

    Now we have to show that this $\Sigma$-type structure is preserved under substitution, i.e. for each $s : \emptyctx.\Delta \to \emptyctx.\Gamma$, we have 
    $
    \cwfsubst{\cwffont{\SigmaCwf{A}{B}}}{s} = \cwffont{\SigmaCwf{\cwfsubst{A}{s}}{\cwfsubst{B}{\cwffont{s.A}}}}
    $
    and 
    $
    {\cwffont{s.\cwffont{\SigmaCwf{A}{B}}}} \circ \cwffont{\coerce{\pairCwf{A[s]}{\Uiso{s}}}{\cwfsubst{B}{\cwfextention{s}{A}}}}  = 
    \cwffont{\pairCwf{A}{B}}
    \circ \cwffont{s.A.B},
    $
    where by $\cwffont{s.A}$ we mean the morphism 
    $\cwffont{\langle s \circ p_{A[s]} , q_{A[s]}\rangle} : 
    \cwfextention{\Delta} {A [s]} \to \cwfextention{\Gamma}{A}$
    in $\constructedcwf{\CC}{\U}$ given by the following composition in 
    $\CC$:
    \[\begin{tikzcd}[ampersand replacement=\&, 
        column sep = large]
	{\emptyctx.\el(\Delta).\el(\coerce{A [s]}{\Uiso{s}})} \& {\emptyctx.\el(\Delta).\el(A)[s]} \& { \emptyctx.\el(\Gamma).\el(A)}
	\arrow["{\emptyctx.\el(\Delta).\Univiso{s}{A}}", from=1-1, to=1-2]
	\arrow["{s.\el(A)}", from=1-2, to=1-3]
    \end{tikzcd}\]
    composed with the obvious instances of $\mathsf{pair}$, $\mathsf{proj}$ and 
    $i^\el$;
    and similarly for $\cwffont{(s.A.B)}$.
    The two goals are proved using standard lemmas about comprehension
    categories and the coherence of $\Sigma$-codes given in \cref{def:sigma-coherence} in \cref{sec:coherences}.
    For the complete proof we refer the reader to our formalisation (
        \coqlink{https://anonymous.4open.science/r/fossacs-strict-hott/src/ComprehensionCats/CwfFromCompCatWithUnivSigma.v#L863}
        and \coqlink{https://anonymous.4open.science/r/fossacs-strict-hott/src/ComprehensionCats/CwfFromCompCatWithUnivSigma.v#L1453}).
\end{proof}

\begin{proof}[Proof of \Cref{thm:universe_pi}]
    The type $\cwffont{\PiCwf{A}{B}}$ in $\constructedcwf{\CC}{\U}$ is defined analogously to the $\Sigma$-type, i.e. 
    given by the following $\Pi$-code in the original model
    \[
    \cwffont{\PiCwf{A}{B}} \coloneqq \PiUniverse{{A}}{
    \cwfsubst{{B}}{\pair{\el({{\Gamma}})}{\el({A})} \circ i_{\SigmaUniverse{{\Gamma}}{{A}}}}}.
    \]
    for a suitable ${\Gamma}$, ${A}$ and ${B}$.
    For a CwF-term ${t}$ of type ${B}$, ${\cwffont{\lambda_{A,B}} ({t})}$
    is a CwF-term of $\cwffont{\PiCwf{A}{B}}$ and is given by the following term of $\el({B})$ in $\CC$:
    \[ \coerce{
  \lambda (
    \coerce{ \coerce{t [\Gamma.i_{\Sigma^\U}^{-1}]
      }{i^{\el}}[\pairCompCat{\el({{\Gamma}})}{\el({{A}})}]}{i^{\el}})}{(i_{\Pi^\U})^{-1}}
    \]
    where we don't include the annotations of the isomorphisms and $\lambda$ for readability.
    Application in $\constructedcwf{\CC}{\U}$ for a CwF-term ${f}$ of $\cwffont{\PiCwf{A}{B}}$ is 
    defined similarly in terms of the application in $\CC$ as follows:
    \[
    \coerce{
    \coerce{
        \coerce{
        \coerce{
            \mathsf{app}(\coerce{f}{i_{\Pi^\U}})
        }{i^\U}
        }{(i^{\el})^{-1}}
        [\projCompCat{\el({{\Gamma}})}{\el({{A}})} \circ \Gamma.i_{\Sigma^\U}]}
        {i^{\mathsf{comp}}}}{{i^{\id{}}}}
    \]
    where $i^{\id{}}$ and $i^\mathsf{comp}$ are the usual coherence isomorphisms corresponding to substitution along identity
    and composition in a comprehension category. 
    For the complete definitions of $\lambda$ and $\mathsf{app}$ we refer the reader to our formalisation.
    It is straightforward to prove the beta and eta rules using standard lemmas about the isomorphisms involved. 

    For the preservation of $\cwffont{\PiCwf{A}{B}}$ under substitution, we have to show  
    $
    \cwfsubst{\cwffont{\PiCwf{A}{B}}}{s} = \cwffont{\PiCwf{\cwfsubst{A}{s}}{\cwfsubst{B}{\cwffont{s.A}}}}
    $ and 
    \[
    \cwfsubst{\cwffont{\lambda_{{A},{B}}} ({t})}{s} 
    = \cwffont{\lambda_{\cwfsubst{{A}}{s}, \cwfsubst{{B}}{\cwffont{s.A}}}} (\cwfsubst{{t}}{\cwffont{s.A}})
    \]
    for each $s : \emptyctx.\Delta \to \emptyctx.\Gamma$ and ${t} : \tmincompcat{}{\el ({B})}$,
    where $\cwffont{s.A}$ is defined the same as in \cref{thm:universe_sigma}.
    The goals are proved using standard lemmas about comprehension
    categories and the coherence of $\Pi$-codes given in \cref{def:pi-coherence} in \cref{sec:coherences}.
    For the complete proof we refer the reader to our formalisation (\coqlink{https://anonymous.4open.science/r/fossacs-strict-hott/src/ComprehensionCats/CwfFromCompCatWithUnivPi.v#L197} 
    and \coqlink{https://anonymous.4open.science/r/fossacs-strict-hott/src/ComprehensionCats/CwfFromCompCatWithUnivPi.v#L222}).
\end{proof}

\end{document}